\documentclass{llncs}

\usepackage{makeidx}  
\usepackage{graphicx}
\usepackage{amsmath}
\usepackage{amssymb}
\usepackage{alltt}

\usepackage{url}
\usepackage{hyperref}

\begin{document}

\mainmatter              

\title{Decidability of an Expressive Description Logic with Rational Grading}

\titlerunning{Rational Grading in an Expressive Description Logic}

\author{Mitko Yanchev}

\authorrunning{M.~Yanchev} 

\institute{Faculty of Mathematics and Informatics, Sofia University
`St. Kliment Ohridski'\\ Bulgaria\\
\email{yanchev@fmi.uni-sofia.bg}}

\maketitle              

\begin{abstract}

In this paper syntactic objects---concept constructors called
\emph{part restrictions} which realize rational grading, are
considered in Description Logics (DLs). Being able to convey
statements about a rational part of a set of successors, part
restrictions essentially enrich the expressive capabilities of
DLs. We examine an extension of well-studied DL
$\mathcal{ALCQIH}_{R^+}$ with part restrictions, and prove that
the reasoning in the extended logic is still decidable. The proof
uses tableaux technique augmented with \emph{indices technique},
designed for dealing with part restrictions.


\end{abstract}


\section{Introduction}

Description Logics (DLs) are widely used in knowledge-based
systems. The representation in the language of transitive
relations, in different possible ways \cite{Sat96}, is important
for dealing with complex objects. \emph{Transitive roles} permit
such objects to be described by referring to their components, or
ingredients without specifying a particular level of
decomposition. The expressive power can be strengthened by
allowing additionally \emph{role hierarchies}. The DL
$\mathcal{ALCH}_{R^+}$ \cite{HG97}, an extension of well-known DL
$\mathcal{ALC}$ with both transitive roles and role hierarchies,
is shown to be suitable for implementation. Though having the same
EXPTIME-complete worst-case reasoning complexity as other DLs with
comparable expressivity, it is more amendable to optimization
\cite{Ho97}.

\emph{Inverse roles} enable the language to describe both the
whole by means of its components and vice versa, for example
$\textsf{has\_part}$ and $\textsf{is\_part\_of}$. This syntax
extension is captured in DL $\mathcal{ALCIH}_{R^+}$ \cite{HST99}.
As a next step, in \cite{HST99} the language is enriched with the
counting (or \emph{grading}---a term coming from the modal
counterparts of DLs \cite{KF72}) \emph{qualifying number
restrictions} what results in DL $\mathcal{ALCQIH}_{R^+}$. It is
given a sound and complete decision procedure for that logic.

We go further considering concept constructors which we call
\textit{part restrictions}, capable of distinguishing a rational
part of a set of successors. These constructors are analogues of
the modal operators for \emph{rational grading} \cite{TY06} which
generalize the majority operators \cite{PS04}. They are $MrR.C$
and (the dual) $WrR.C$, where $r$ is a rational number in
(0,\thinspace 1), $R$ is a role, and $C$ is a concept. The
intended meaning of $MrR.C$ is `(strongly) $M$ore than $r$-part of
$R$-successors (or $R$-\emph{neighbours}, in the presence of
inverse roles) of the current object possess the property C'. Part
restrictions essentially enrich the expressive capabilities of
DLs. From the `object domain' point of view they seem to be more
`socially' than `technically' oriented, but in any case they give
new strength to the language. The usual example of the use of part
restrictions is to express the notion of qualifying majority in a
voting system: $\textsf{M}{2\over 3}\thinspace
\textsf{voted}.\textsf{Yes}$.

On the other hand, \emph{presburger constraints} in the language
of extended modal logic EXML
\nobreak{\cite{Dem06}}, a many-relational language with only
independent relations, capture both integer and rational grading,
and have rich expressiveness. The rational grading modal operators
are expressible by the presburger constraints, and the
satisfiability of EXML is shown to be in PSPACE. Another
constraints on role successors witch subsume the standing alone
part restrictions are introduced in \cite{Baa_ALCSCC} using the
quantifier-free fragment of Boolean Algebra with Presburger
Arithmetic. The corresponding DL $\mathcal{ALCSCC}$ also captures
both integer and rational grading, and it is shown that the
complexity of reasoning in it is the same as in $\mathcal{ALCQ}$
(the DL extending $\mathcal{ALC}$ with qualifying number
restrictions), both without and with TBoxes.

A combinatorial approach to grading in modal logics, uses the so
called \emph{majority digraphs}
\nobreak{\cite{Moss_most-graph}}. In this
approach, in addition to integer and rational grading, also the
grading with real coefficients can be expressed.

Nonetheless, the use of separate rational grading, having its
place also in modal logics, proves markedly beneficial in DLs.
Part restrictions can be combined in a DL with many other
constructors. \emph{Indices technique}, specially designed for
exploring the part restrictions, allows following a common way for
obtaining decidability and complexity results as in less, so in
more expressive languages with rational grading. In particular,
reasoning complexity results---polynomial, NP, and
co-NP---concerning a range of description logics from the
$\mathcal{AL}$-family with part restrictions added, are obtained
(\cite{Y-LC2012}, \cite{Y-CiE2012}, \cite{Y-CiE2013}), as well as
PSPACE-results for modal and expressive description logics
(\cite{Y-LC2014}, \cite{Y-CiE2014}, \cite{Y-PLS2017}).

Now we consider the DL $\mathcal{ALCQPIH}_{R^+}$, in which the
language of $\mathcal{ALCQIH}_{R^+}$ is augmented with part
restrictions. We use the tableaux technique to prove that the
reasoning in the extended logic is still decidable.

\section{Syntax and Semantics of $\mathcal{ALCQPIH}_{R^+}$}

The $\mathcal{ALCQPIH}_{R^+}$-syntax and semantics differ from
those of $\mathcal{ALCQIH}_{R^+}$ only in the presence of part
restrictions. 

\begin{definition} {Let $\emph{\textbf{C}}_\emph{o}\neq\emptyset$ be a set of
\emph{concept names}, $\emph{\textbf{R}}_\emph{o}\neq\emptyset$ be
a set of \emph{role names}, some of which transitive, and
$\emph{\textbf{Q}}_0$ be a set of rational numbers in $(0,1)$. We
denote the set of transitive role names $\emph{\textbf{R}}^+$, so
that $\emph{\textbf{R}}^+\subseteq\emph{\textbf{R}}_\emph{o}$.
Then we define the set of $\mathcal{ALCQPIH}_{R^+}$-\emph{roles}
(we will refer to simply as `roles') as
$\emph{\textbf{R}}=\emph{\textbf{R}}_\emph{o} \cup \{R^-\ |\ R \in
\emph{\textbf{R}}_\emph{o}\}$, where $R^-$ is the inverse role of
$R$.}

{As the inverse relation on roles is symmetric, to avoid
considering roles such as $R^{--}$ we define a function
$\emph{\textsf{Inv}}$ which returns the inverse of a role.
Formally, $\emph{\textsf{Inv}}(R)=R^-$, if $R$ is a role name, and
$\emph{\textsf{Inv}}(R^-)=R$. Thus,
$\emph{\textsf{Inv}}(\emph{\textsf{Inv}}(R))=R$.}

{A \emph{role inclusion axiom} has the form $R\sqsubseteq S$, for
two roles $R$ and $S$, and the acyclic inclusion relation
$\sqsubseteq$. For a set of role inclusion axioms $\mathcal{R}$, a
\emph{role hierarchy} is $\mathcal{R}^+:= \bigl(\mathcal{R}\cup
\{\emph{\textsf{Inv}}(R)\sqsubseteq \emph{\textsf{Inv}}(S)\ |\
R\sqsubseteq S \in \mathcal{R}\},\ \sqsubseteq^+ \bigr)$, where
$\sqsubseteq^+$ is the reflexive and transitive closure of
$\sqsubseteq$ over $\mathcal{R}\cup
\{\emph{\textsf{Inv}}(R)\sqsubseteq \emph{\textsf{Inv}}(S)\ |\
R\sqsubseteq S \in \mathcal{R}\}$.}



A role $R$ is \emph{simple} with respect to $\mathcal{R}^+$ iff
$R\not\in \emph{\textbf{R}}^+$ and, for any $S\sqsubseteq^+ \!\!
R,\ S$ is also simple w.r.t. $\mathcal{R}^+$.

The set of $\mathcal{ALCQPIH}_{R^+}$\!-\emph{concepts} (we will
refer to simply as `concepts') is the smallest set such that: 1.
every concept name is a concept; 2. if $C$ and $D$ are concepts,
and $R$ is a role, then $\neg C,\ C\sqcap D,\ C\sqcup D,\  \forall
R.C$, and $\exists R.C$ are concepts; 3. if $C$ is a concept, $R$
is a simple role, $n\geq 0$, and $r\in \emph{\textbf{Q}}_0$, then
$\geqslant n R.C,\ \leqslant n R.C,\ MrR.C$, and $WrR.C$ are
concepts.
\end{definition}

The limitation roles in qualifying number restrictions, as well as
in part restrictions to be simple is used essentially in the
proofs. From the other side, the presence in the language of role
hierarchies together with only number restrictions on transitive
roles leads to undecidability \cite{HST99_practical}\smallskip

An \emph{interpretation} $\mathcal{I} = (\Delta^\mathcal{I},
\cdot^\mathcal{I})$ consisting of a nonempty set
$\Delta^\mathcal{I}$, called the \emph{domain} of $\mathcal{I}$,
and a function $\cdot^\mathcal{I}$ which maps every concept to a
subset of $\Delta^\mathcal{I}$ and every role to a subset of
$\Delta^\mathcal{I}\!\!\times\!\Delta^\mathcal{I}$, is defined in
a standard way.\footnote{All definitions and techniques from
Section 5 of \cite{HST99} concerning $\mathcal{ALCQIH}_{R^+}$ are
applicable to the extended DL, eventually with only mild changes.
So, in what follows we present explicitly, due to the restriction
of space, only what is new, or changed, relying on and referring
to \cite{HST99} for the rest. The complete definition of the
interpretation, the complete sets of tableaux properties and
completion rules, and all proofs can be seen in the appendix.}
We only set the additional restriction for any object
$x\in\Delta^\mathcal{I}$ and any role $R\in \textbf{R}$ the set of
objects, $R^\mathcal{I}$-related to
($R^\mathcal{I}$-\emph{neighbours} of) $x$, denoted
$R^\mathcal{I}(x)$, to be finite. $R^\mathcal{I}(x, C)$ denotes
the set $\{y\ |\ \langle x,y\rangle\in R^\mathcal{I}$ and $y\in
C^\mathcal{I}\}$ of $R^\mathcal{I}$-neighbours of $x$ which are in
$ C^\mathcal{I}$, and $\sharp M$ denotes the cardinality of a set
$M$.
For part restrictions, for any concept $C$, simple role $R$, and
$r\in \textbf{Q}_0$ the definitions of mapping are:
\medskip


$\phantom{12} (MrR. C)^\mathcal{I} = \{ x\in \Delta^\mathcal{I}\
|\ \sharp R^\mathcal{I}(x, C) > r.\sharp R^\mathcal{I}(x)\}$
\smallskip

$\phantom{12} (WrR. C)^\mathcal{I} = \{ x\in \Delta^\mathcal{I}\
|\ \sharp R^\mathcal{I}(x,\neg C) \leq r.\sharp
R^\mathcal{I}(x)\}\ \Big( = (\neg MrR.\neg C)^\mathcal{I}
\Big)$\medskip



An interpretation $\mathcal{I}$ \emph{satisfies a role hierarchy}
$\mathcal{R}^+$ iff $R^\mathcal{I}\subseteq S^\mathcal{I}$ for any
$R\sqsubseteq^+ \!\! S \in \mathcal{R}^+$; we denote that by
$\mathcal{I}\models \mathcal{R}^+$.

A concept $C$ is \emph{satisfiable} with respect to a role
hierarchy $\mathcal{R}^+$ iff there exists an interpretation
$\mathcal{I}$ such that $\mathcal{I}\models \mathcal{R}^+$ and
$C^\mathcal{I}\not = \emptyset$. Such an interpretation is called
a \emph{model of $C$ with respect to} $\mathcal{R}^+$. For an
object $x \in C^\mathcal{I}$ we say that $x$ \emph{satisfies} $C$,
also that $x$ is an \emph{instance} of $C$, while $x \in
\Delta^\mathcal{I}\backslash C^\mathcal{I}$ \emph{refuses} $C$.

Thus, for $x\in \Delta^\mathcal{I},\ x$ is in
$(MrR.C)^\mathcal{I}$ iff strictly greater than $r$ part of
$R^\mathcal{I}$-neighbours of $x$ satisfies $C$, and $x$ is in
$(WrR. C)^\mathcal{I}$ iff no greater than $r$ part of
$R^\mathcal{I}$-neighbours of $x$ refuses $C$.

A concept $D$ \emph{subsumes} a concept $C$ with respect to
$\mathcal{R}^+$ (denoted $C\sqsubseteq_{\mathcal{R}^+} \!\! D$)
iff $C^\mathcal{I}\subseteq D^\mathcal{I}$ holds for every
interpretation $\mathcal{I}$ such that
$\mathcal{I}\models\mathcal{R}^+$.

Checking the subsumption between concepts is the most general
reasoning task in DLs. From the other side, $C\! \sqsubseteq\! D$
iff $C \sqcap~\!\neg D$ is unsatisfiable. Thus, in the presence of
negation of an arbitrary concept, checking the (un)satisfiability
becomes as complex as checking the subsumption.


In what follows we consider concepts to be in the \textit{negation
normal form} (NNF). We denote the NNF of $\neg C$ by $\sim\! C$.
The NNF of $\sim\! MrR.C$ is $WrR.\neg C$, and, dually, $\sim\!
WrR.C = MrR.\neg C$. For any concept $C$ in NNF we denote with
$clos(C)$ the smallest set of concepts containing $C$ and closed
under sub-concepts and $\sim$. The size of $clos(C)$ is linear to
the size of $C$. With $\textbf{R}_C$ we denote the set of roles
occurring in $C$ and their inverses.

\section{A Tableau for $\mathcal{ALCQPIH}_{R^+}$}


We will use a tableaux algorithm to test the satisfiability of a
concept. We extend the definition of
$\mathcal{ALCQIH}_{R^+}$-tableau by modifying one property to
reflect the presence of part restrictions, and adding two new
ones. Thus we obtain a definition of a tableau for
$\mathcal{ALCQPIH}_{R^+}$.

\begin{definition} A \emph{tableau} $T$ for a concept $D$ in NNF
with respect to a role hierarchy $\mathcal{R}^+$ is a triple
($\emph{\textbf{S}}, \mathcal{L}, \mathcal{E}$), where
$\emph{\textbf{S}}$ is a set of individuals, $\mathcal{L}:
\emph{\textbf{S}} \rightarrow 2^{clos(D)}$ is a function mapping
each individual of $\emph{\textbf{S}}$ to a set of concepts which
is a subset of $clos(D)$, $\mathcal{E}: \emph{\textbf{R}}_D
\rightarrow 2^{\emph{\textbf{S}}\times \emph{\textbf{S}}}$ is a
function mapping each role occurring in $\emph{\textbf{R}}_D$ to a
set of pairs of individuals, and there is some individual $s\in
\emph{\textbf{S}}$ such that $D\in \mathcal{L}(s)$. For all
individuals $s,t\in \emph{\textbf{S}}$, concepts in $clos(D)$, and
roles in $\emph{\textbf{R}}_D$, $T$ must satisfy 13 properties.
%
\end{definition}


We denote with $R^T(s)$ the set of individuals, $R$-related to
$s$, and $R^T(s, C):=\{t\in {\textbf{S}}\ |\ \langle s,t\rangle\in
\mathcal{E}(R)\ \textrm{and}\ C \in \mathcal{L}(t)\}$. The new and
the modified properties follow. In property 13 (modified property
11 from the definition of $\mathcal{ALCQIH}_{R^+}$-tableau), and
in what follows, $\boxtimes$ is a placeholder, besides for
$\geqslant n$ and $\leqslant n$, for arbitrary $n\geq 0$, also for
$\exists$, and for $Mr$ and $Wr$, for arbitrary $r\in
\textbf{Q}_0$.\medskip

$11.\ \textrm{If}\ MrR.C\in \mathcal{L}(s),\ \textrm{then}\ \sharp
R^T(s,C)>r.\sharp R^T(s).$\smallskip

$12.\ \textrm{If}\ WrR.C\in \mathcal{L}(s),\ \textrm{then}\ \sharp
R^T(s,\sim\! C)\leq r.\sharp R^T(s).$\smallskip

$13.\ \textrm{If}\ \boxtimes R.C\in \mathcal{L}(s)\ \textrm{and}\
\langle s,t\rangle \in \ \mathcal{E}(R),\ \textrm{then}\ C\in
\mathcal{L}(t)\textrm{ or } \sim\! C\in\mathcal{L}(t).$\medskip

Having the definition of $\mathcal{ALCQPIH}_{R^+}$-tableau, we can
prove Lemma \ref{lemma:I-Ta} following the standard way, also for
the new and modified properties.

\begin{lemma}\label{lemma:I-Ta}
An $\mathcal{ALCQPIH}_{R^+}$-concept D is satisfiable with respect
to a role hierarchy $\mathcal{R}^+$ iff there exists a tableau for
D with respect to $\mathcal{R}^+$.
\end{lemma}

\section{Constructing an $\mathcal{ALCQPIH}_{R^+}$-Tableau}

Lemma \ref{lemma:I-Ta} guarantees that the algorithm constructing
tableaux for $\mathcal{ALCQPIH}_{R^+}$-concepts can serve as a
decision procedure for concept satisfiability (and hence, also for
subsumption between concepts) with respect to a role hierarchy
$\mathcal{R}^+$. We present such an algorithm.

As usual with the tableaux algorithms,
$\mathcal{ALCQPIH}_{R^+}$-algorithm tries to prove the
satisfiability of a concept $D$ by constructing a \emph{completion
tree} (c.t. for short) $\textbf{T}$, from which a tableau for $D$
can be build. Each node $x$ of the tree is labelled with a set of
concepts $\mathcal{L}(x)$ which is a subset of $clos(D)$, and each
edge $\langle x,y \rangle$ is labelled with a set of roles
$\mathcal{L}(\langle x,y \rangle)$ which is a subset of
$\textbf{R}_D$. The algorithm starts with a single node (the c.t.
root) $x_0$ with $\mathcal{L}(x_0)=\{D\}$, and the tree is then
expanded by completion rules, which decompose the concepts in the
nodes' labels, and add new nodes and edges, giving the
relationships between nodes, and new labels to the nodes and
edges.

A node $y$ is an $R$-\emph{successor} of a node $x$ if $y$ is a
successor of $x$ and $S\in \mathcal{L}(\langle x,y \rangle)$ for
some $S$ with $S\sqsubseteq^+R$; $y$ is an $R$-\emph{neighbour} of
$x$ if it is an $R$-successor of $x$,
or if $x$ is an $\textsf{Inv}(R)$-successor of $y$.

We denote with $R^\textbf{T}(x)$ the set of $R$-neighbours of a
node $x$ in the c.t. $\textbf{T}$, and with $R^\textbf{T}(x,
C)$---the set of $R$-neighbours of $x$ in $\textbf{T}$ which are
labelled with $C$.



A c.t. $\textbf{T}$ is said to contain a \emph{clash} (i.e., the
obvious contradiction) if, for some node $x$ in $\textbf{T}$, a
concept $C$, a role $R$, some $n\geq0$, and some $r\in
\textbf{Q}_0$ any of the following is the case. Otherwise it is
\textit{clash-free}.\smallskip

$\phantom{123} CL1.\ \{C,\thinspace \neg C\}\subseteq
\mathcal{L}(x)$ \smallskip

$\phantom{123} CL2.\ \leqslant n R. C\in \mathcal{L}(x)$ and
$\sharp R^\textbf{T}(x,C) > n$ \smallskip

$\phantom{123} CL3.\ MrR.C\in \mathcal{L}(x) \text{ and } \sharp
R^\textbf{T}(x,C)\leq r.\sharp R^\textbf{T}(x)$\smallskip

$\phantom{123} CL4.\ WrR.C\in \mathcal{L}(x) \text{ and } \sharp
R^\textbf{T}(x,\sim\! C)>r.\sharp R^\textbf{T}(x)$\smallskip



A completion tree is \emph{complete} if none of the completion
rules is applicable, or if for some node $x$, $\mathcal{L}(x)$
contains a clash of type $CL1$ or type $CL2$.\footnote{Part
restrictions talk about no exact quantities, but ratios. So,
instances of $CL3$ and $CL4$ (which are also conditions for
applicability of $M$-rule and $W$-rule, see
Figure~\ref{figure:rules}) can appear and disappear dynamically
during the c.t. generation. That is why we exclude them from the
definition of the c.t. completeness.}

If, for a concept $D$, the completion rules can be applied in a
way to yield a complete and clash-free completion tree, then the
algorithm returns `$D$ \emph{is satisfiable}'; otherwise, it
returns `$D$ \emph{is unsatisfiable}'. \smallskip

During the expansion the algorithm uses the \emph{pair-wise
blocking} technique as defined in \cite{HST99}, Sections 4.1 and
5.3, to ensure only finite
paths in the completion tree. It also uses \emph{indices
technique} which will be presented in details, to prevent from
infinite branching of the tree (possibly) caused by part
restrictions.

Figure~\ref{figure:rules} presents the completion rules which are
new or modified in comparison with ones in the
$\mathcal{ALCQIH}_{R^+}$-algorithm. \emph{choose-rule} is
augmented (via the placeholder $\boxtimes$) to add also labels,
induced by $\exists$-concepts and part restrictions.

In the presence of part restrictions, $\geqslant$-\emph{rule}
which adds all the necessary successors at ones leads to
incompleteness.\footnote{For example, the concept $A\sqcap\exists
R^-.\big( \!\!\geqslant\! 4R.\top\sqcap\leqslant\!5R.\top\sqcap
M{2 \over 5}R.(\neg A) \sqcap W{1\over 2}R.A\big)$, where $\top =
A\sqcup\neg A$, and $A$ is a concept name, has a unique tableau
(modulo labelling of the individuals) with just four individuals.
But no complete and clash-free c.t. can be built from that tableau
using the `all-at-once' $\geqslant$-\emph{rule}.} So, it is
modified to add successors one by one, thus preventing the
occurrence of redundant neighbours. This needs some modification
of $\leqslant$-\emph{rule} also. $\leqslant$-rule transfers the
label of an edge to just one other edge. So, if two edges are
labelled with the same role, it has been labelling initially (even
if some label transfer has happened meanwhile) two different edges
connecting $x$ with two of its neighbours. The possible cases are:
1) the (concept) labels of the neighbours $y$ and $z$ are
different and contradict each other for any relabelling by
$choose$-rule, so $y$ and $z$ cannot be merged; 2) the labels of
$y$ and $z$ are different but there is a labelling by
$choose$-rule which makes them not contradicting, then there is no
need nodes to be merged, as when this labelling is made to the
firstly generated node, the second one would not be generated at
all; 3) the labels of $y$ and $z$ are identical, then the
generation of both nodes is triggered by $\geqslant\!\!n$-concept
with $n\geq 2$, $M$-, or $W$-concept, or, anyway, they are used
for the c.t.-satisfying of such a concept, so they \emph{must not}
be merged. This justifies the use of $\mathcal{L}(\langle
x,y\rangle)\cap\mathcal{L}(\langle x,z\rangle)=\emptyset$
condition in the rule.

$M$-\emph{rule} and $W$-\emph{rule} (the \emph{part rules}) are
new \emph{generating rules} (in addition to $\exists$-rule and
$\geqslant$-rule) which deal with part restrictions. The rest of
the rules---$\sqcap$-, $\sqcup$-, $\exists$-, $\forall$-, and
$\forall_+$-rule---remain just as they are in \cite{HST99},
Section 5, Figure 5.

\begin{figure}
\begin{tabular}{l l}
${choose}$- &\hspace{4mm} If 1. $\boxtimes R.C\in \mathcal{L}(x)$,
$x$ is not indirectly blocked, and\\
rule: &\hspace{7.3mm} 2. there is an $R$-neighbour $y$ of $x$
with $\{C,\sim C\}\cap\mathcal{L}(y)=\emptyset$\\
&\hspace{4mm} then $\mathcal{L}(y)\rightarrow\mathcal{L}(y)\cup\{E
\}$ for some $E\in\{C,\sim C\}$\smallskip
\end{tabular}

\begin{tabular}{l l}
$\geqslant$-rule: &\hspace{4.5mm} If 1. $\geqslant nR.C\in
\mathcal{L}(x)$,
$x$ is not blocked, and\\
&\hspace{7.7mm} 2. $\sharp  R^\textbf{T}(x,C) < n$\\
&\hspace{4.5mm} then create a new successor $y$ of $x$ with
$\mathcal{L}(\langle x,y\rangle)=\{ R\}$ and
$\mathcal{L}(y)=\{C\}$ \smallskip
\end{tabular}

\begin{tabular}{l l}
$\leqslant$-rule: &\hspace{4.5mm} If 1. $\leqslant nR.C\in
\mathcal{L}(x)$,
$x$ is not indirectly blocked, and\\
&\hspace{7.7mm} 2. $\sharp  R^\textbf{T}(x,C) > n$ and there are
two $R$-neighbours $y$ and $z$ of $x$ with\\
&\hspace{11.5mm} $C\in\mathcal{L}(y)$, $C\in\mathcal{L}(z)$,
$\mathcal{L}(\langle x,y\rangle)\cap\mathcal{L}(\langle x,z\rangle)=\emptyset$, and\\
&\hspace{11.5mm} $y$ is not a predecessor of $x$ \\
&\hspace{4.5mm} then 1.
$\mathcal{L}(z)\rightarrow\mathcal{L}(z)\cup\mathcal{L}(y)$,\\
&\hspace{12.1mm} 2. If $z$ is a predecessor of $x$\\
&\hspace{16.1mm} then $\mathcal{L}(\langle
z,x\rangle)\rightarrow\mathcal{L}(\langle z,x\rangle)\cup
\{\textsf{Inv}(S)|S\in\mathcal{L}(\langle x,y\rangle)\}$\\
&\hspace{16.1mm} else $\mathcal{L}(\langle
x,z\rangle)\rightarrow\mathcal{L}(\langle x,z\rangle)\cup
\mathcal{L}(\langle x,y\rangle)$\\
&\hspace{12.1mm} 3. $\mathcal{L}(\langle
x,y\rangle)\rightarrow\emptyset$ \smallskip
\end{tabular}

\begin{tabular}{l l}
${M}$-rule: &\hspace{3.5mm} If 1. $MrR.C\in \mathcal{L}(x)$, $x$
is not blocked, and\\
&\hspace{6.7mm} 2. $\sharp  R^\textbf{T}(x,C)\leq r.\sharp
R^\textbf{T}(x)$ then calculate $BAN_x$, and if\\
&\hspace{6.7mm} 3. $\sharp  R^\textbf{T}(x)<BAN_x$\\
&\hspace{3.5mm} then create a new successor $y$ of $x$ with
$\mathcal{L}(\langle x,y\rangle)=\{ R\}$ and $\mathcal{L}(y)=\{
C\}$\smallskip
\end{tabular}

\begin{tabular}{l l}
${W}$-rule: &\hspace{3.5mm} If 1. $WrR.C\in \mathcal{L}(x)$, $x$
is not blocked, and\\
&\hspace{6.7mm} 2. $\sharp  R^\textbf{T}(x,\sim\! C) > r.\sharp
R^\textbf{T}(x)$ then calculate $BAN_x$, and if\\
&\hspace{6.7mm} 3. $\sharp  R^\textbf{T}(x)<BAN_x$\\
&\hspace{3.5mm} then create a new successor $y$ of $x$ with
$\mathcal{L}(\langle x,y\rangle)=\{ R\}$ and $\mathcal{L}(y)=\{
C\}$
\end{tabular}

    \caption{The new and the modified completion rules}\label{figure:rules}
\end{figure}

We impose a \emph{rule application strategy} any generating rule
to be applied only if all \emph{non-generating rules} (i.e.,
$\sqcap$-, $\sqcup$-, $\forall$-, $\forall_+$-, ${choose}$- and
$\leqslant$-rule) are inapplicable. Apart from that the generation
process is non-deterministic in both which rule (in any group---of
non-generating and generating ones) to be applied, and which
concept(s) to be chosen in the non-deterministic $\sqcup$-,
${choose}$-, and $\leqslant$-rule.

The rule application strategy is essential for the successful
`work' of $\leqslant$-rule, and for part rules. It ensures that
a)~all concepts `talking' about neighbours are already present in
$\mathcal{L}(x)$, and b)~all possible (re)labelling of
neighbours of $x$ is done before the application of a part rule. 
Both are necessary for applying the indices technique for the
correct generation of successors, caused by part restrictions. The
check-up in part rules (in 3.) for not reaching the \emph{border
amount of neighbours} for the current node $x$ ($BAN_x$) is a kind
of `horizontal blocking' of the generation process, used to ensure
inapplicability of part rules after a given moment. The notion is
crucial for the termination of the algorithm, and its use is based
on Lemma~\ref{lemma:BAN}, which is the upshot of the indices
technique.

\section{Indices Technique}


We develop a specific technique, which we call \emph{indices
technique}, to cope with the presence of part restrictions. This
technique permits to extend appropriately the definition of a
clash, to design completion rules, dealing with part restriction,
and to give an adequate rule application strategy, as they are
presented in the previous section, all to guarantee the
correctness of the tableaux algorithm.

\subsection{The clashes with part restrictions}\label{CL}


$CL3$ and $CL4$, which are also conditions for applicability of
part rules, are dynamic. Applied consecutively, part rules can
`repair' one clash, and, at the same time, provoke another. Thus,
instances of $CL3$ and $CL4$ can appear and disappear, in some
cases infinitely, during the c.t. generation, even if the initial
concept is satisfiable. So, we have to take special care both to
ensure the termination of part rules application, and not to leave
avoidable `part' clashes in the completion tree. That turns out to
be the main difficulty in designing the algorithm. We overcome it
by proving that if it is possible to unfold part restrictions at a
given node avoiding \emph{simultaneously} both kinds of clashes,
it can be done within some number of neighbours. As clashes are
always connected with a single node, talking about its label and
its neighbours, that is enough to guarantee the termination. The
following subsection presents the technique in details.

%

\subsection{Counteracting part restrictions. Clusters}\label{Clusters}

We start our analysis with the simplest case when, for a node $x$
of the c.t. $\textrm{\textbf{T}}$, there are in $\mathcal{L}(x)$
only part restrictions, and they all are with the same role $R$,
and with sub-concepts which are either a fixed concept $C$, or its
negation $\sim\!C$, and $x$ is not an $\textsf{Inv}(R)$-successor.
All such part restrictions form the set:
\smallskip

$\phantom{1234567890} \{M{r_1}R.C,\ M{r_2}R.\!\sim\!C,\
W{r_3}R.C,\ W{r_4}R.\!\sim\!C\}\ \quad (1)$ \smallskip


We call the subset of ($1$) which is in $\mathcal{L}(x)$ a
\emph{cluster of $R$ and $C$ at $x$ in $\emph{\textbf{T}}$}, and
we denote it $Cl_x^\textbf{T}(R,C)$. It is obvious, that during
the generation of ($R$-) successors of $x$ (if it is necessary)
instances of $CL3$ and $CL4$ can appear only if two contradicting
part restrictions are in that cluster.


\begin{definition} {A part
restriction which is in the label of a node $x$ in a c.t.
$\emph{\textbf{T}}$ is \emph{${{\textbf{T}}}$-satisfied} (at $x$)
if there is no clash with it at $x$.\\
\indent A {cluster is $\emph{\textbf{T}}$-\emph{satisfied}} if all
part restrictions in it
are $\emph{\textbf{T}}$-satisfied.\\
\indent A {cluster is \emph{c.t.-satisfiable}} if it can be
$\emph{\textbf{T}}$-satisfied, for some c.t. $\emph{\textbf{T}}$.}
\end{definition}

In fact, in ($1$) there can be more than one part restriction of
any of the four types. But note that, if $MrR.C$ is
$\textrm{\textbf{T}}$-satisfied, then that is the case with
$M{r'}R.C$ (being in the label of the same node), for any $r' <
r$. So, we can take $r_1$ and $r_2$ to be the maximums, and, by
analogues reasons, $r_3$ and $r_4$ to be the minimums of the $r$-s
in part restrictions of the corresponding types. Thus we obtain
the upper, representative for the c.t.-satisfiability of all part
restrictions in the label of a node, set with only four ones.

The idea behind c.t.-satisfiability is that if a cluster, and more
general, the set of all part restrictions labelling a given node,
is c.t.-satisfiable, then a c.t. without clashes with part
restrictions at that node can be non-deterministically generated,
while the part rules become inapplicable for this node (as
inequality in condition 2 in part rules becomes false). So,
concerning part restrictions, c.t.-satisfiability is a sufficient
condition for obtaining a clash-free complete c.t.




Our next observation is that both $M{r_1}R.C$ and $W{r_3}R.C$ act
in the same direction concerning c.t.~generation, as the former
forces the addition of enough $R$-successors of $x$ labelled with
$C$, and the latter limits the number of $R$-successors of $x$
labelled with $\sim\!C$. The same holds for $M{r_2}R.\!\sim\!C$
and $W{r_4}R.\!\!\sim\!\!C$ with respect to $\sim\!C$. At that
time, as $M{r_1}R.C$, so $W{r_3}R.C$ counteract with any of
$M{r_2}R.\!\!\sim\!\!C$ and $W{r_4}R.\!\!\sim\!\!C$. This leads to
two main possibilities for $Cl_x^\textbf{T}(R,C)$: 

A. The cluster contains only part restrictions, acting in the same
direction (or just a single one)---we call it \textit{cluster of
type} A, or A-cluster. In the absence of counteracting part
restrictions these clusters are always c.t.-satisfiable.

B. The cluster contains at least two counteracting part
restrictions---we call it \textit{cluster of type} B, or
B-cluster.

%



In order the c.t. generation process to be able to c.t.-satisfy a
B-cluster and to avoid $CL3$ and $CL4$ clashes, the next
inequalities between the $r$-s in the cluster (or between the
\emph{indices}, from where we take the name of indices technique)
must be fulfilled---follows directly from the semantics of part
restrictions, the above remarks about counteractions, and the
definition of c.t.-satisfiability:
\smallskip

%
%
%
%
%

\begin{tabular}{l l}
\hspace{15mm} $1^\circ\ r_1+r_2<1$ & \hspace{25mm} $4^\circ\ r_3+r_4\geq 1$, what is\\

\hspace{15mm} $2^\circ\ r_1<r_4$ & \hspace{30mm}(a) $r_3+r_4>1$, or\\

\hspace{15mm} $3^\circ\ r_2<r_3$ & \hspace{30mm}(b) $r_3+r_4=1$
\smallskip
\end{tabular}

If any of the inequalities $1^\circ$--$4^\circ$ does not hold, any
complete c.t. will contain a clash, as it is impossible to
c.t.-satisfy simultaneously (at the same node) the part
restrictions in which are the indices, taking part in the failed
inequality.

We can combine that four inequalities in just one taking into
account the kind of interaction between part restrictions.
$W{r_3}R.C$ means that $\sim\!C$ has to label not greater than
$r_3$ part of all $R$-neighbours of $x$, i.e., that $C$ has to
label at least ($1-r_3$) part of them. We set
$\check{r}=\max\bigl(\{r_1,1-r_3\}\bigr)$
(or, if the part restriction with either $r_1$, or $r_3$ is not in
the cluster, $\check{r}$ is just the expression with the other).
Now, it is obvious that if $C$ labels greater than $\check{r}$
part of all $R$-neighbours of $x$, then both $M{r_1}R.C$ and
$W{r_3}R.C$ are (or the single one from the couple, which is in
the cluster, is) c.t.-satisfied. Analogues reasonings go with the
other couple of part restrictions, acting in the same direction
(the ones with $M{r_2}$ and $W{r_4}$), and (a part smaller than)
$\hat{r}=\min\bigl(\{1-r_2,r_4\}\bigr)$.

We call \textit{dominating} the part restrictions which determine
$\check{r}$ and $\hat{r}$.

It is important to note that $r_3+r_4=1$ does not spoil the
c.t.-satisfiability (unlike $r_1+r_2=1$). We exclude it from the
general examination, as a special sub-case, and discuss it
separately. Thus, case B divides into two sub-cases:

\noindent B(a). The cluster contains no counteracting $W$ part
restrictions, or $r_3+r_4\not =1$.

\noindent B(b). The cluster contains counteracting $W$ part
restrictions and $r_3+r_4=1$. \bigskip

%
%



\noindent \textbf{Clusters of type B(a).} Our first claim is:



\begin{lemma} \label{lemma-r's}
For a B(a)-cluster $Cl_x^\emph{\textbf{T}}(R,C)$, the inequalities
 $1^\circ$, $2^\circ$, $3^\circ$, and
$4^\circ$\emph{(a)}, with the corresponding part restrictions
being in the cluster, hold iff $\check{r}<\hat{r}$. 



\end{lemma}

%
%
%


\begin{corollary} \label{r<-ness} $\check{r}<\hat{r}$
is a necessary condition for the c.t.-satisfiability of a
B(a)-cluster $Cl_x^\emph{\textbf{T}}(R,C)$.
\end{corollary}


The upper inequality is also a \textit{sufficient condition} for a
cluster's c.t.-satisfi\-abi\-li\-ty. Indeed, if
$\check{r}<\hat{r}$, and the number of $R$-neighbours of $x$
labelled with $C$---$|R^\textrm{\textbf{T}}(x, C)|$---is strongly
(due to the strong inequality in the $M$-rule) between
$\check{r}.|R^\textrm{\textbf{T}}(x)|$ and
$\hat{r}.|R^\textrm{\textbf{T}}(x)|$, then the dominating part
restrictions are c.t.-satisfied, and so are the rest of the part
restrictions in the cluster, if any. This shows that
$\check{r}<\hat{r}$ guarantees the c.t.-satisfiability; practical
c.t.-satisfaction of a cluster depends on the number of
neighbours, and, of course, their appropriate labelling.

%

Note also that even though $\check{r}<\hat{r}$ holds, we can have
\emph{instable} c.t.-satisfaction, as can be seen from the next
example. Let the dominating part restrictions be $M{2\over 3}R.C$
and $M{1\over 4}R.\!\!\sim\!C$. They can be
$\textrm{\textbf{T}}$-satisfied if $R^\textrm{\textbf{T}}(x)$ has
10 nodes (with $C$ labelling 7, and $\sim\!\!C$---3 of them), and
also 11 nodes (with labelling
$C:\thinspace\sim\!\!C$---8\thinspace:\thinspace 3), while if
$R^\textrm{\textbf{T}}(x)$ has 12 nodes, there is no way these
part restrictions to be $\textrm{\textbf{T}}$-satisfied, as the
first wants $C$ to label at least 9, and the second---$\sim\!\!C$
to label at least 4 $R$-neighbours of $x$. In case of 13
$R$-neighbours of $x$ the part restrictions again can be
simultaneously $\textrm{\textbf{T}}$-satisfied.

\begin{definition}
{A cluster $Cl_x^{\emph{\textbf{T}}}(R,C)$ is}
$n$-\emph{satisfiable}, {where $n\geq 0$, if it can be
c.t.-satisfied when $x$ has exactly $n\ R$-neighbours.}

{A cluster is} \emph{stably $n$-satisfiable}, {if it is
$n$-satisfiable, and for any natural number $n'>n$ it is also
$n'$-satisfiable.}

{A cluster is} \emph{stably c.t.-satisfiable}, {if it is stably
$n$-satisfiable for some $n\geq 0$.}
\end{definition}

Note that from the above definition it follows that if a cluster
is stably $n$-satisfiable, it is also stably $n'$-satisfiable, for
any natural number $n'>n$.

In the example above the cluster is 10-, and 11-satisfiable, it is
not 12-satisfiable, and it is (in fact---stably) 13-satisfiable.

So, if we have a sufficient condition for stable n-satisfiability
of B(a)-clusters, we will know exactly when, in the
non-deterministic c.t. generation process, stable
c.t.-satisfaction of such a cluster will be achieved in at least
one non-determinis\-tic generation (we call it a \emph{successful}
generation). Then we will be able to key at that moment the part
rules with respect to the part restrictions of that cluster, thus
avoiding infinite rules application in the unsuccessful
generations.

\begin{lemma} \label{lemma_*-sat}
Let, for a B(a)-cluster $Cl_x^\emph{\textbf{T}}(R,C)$,
$\check{r}<\hat{r}$ hold. Then a sufficient condition for the
non-deterministic $|R^\emph{\textbf{T}}(x)|$-satisfiability of the
cluster is: \medskip

$\phantom{1234567890123456789012} |R^\emph{\textbf{T}}(x)| >
{1\over
{\hat{r}-\check{r}}} \quad (\sharp)$


\end{lemma}
%
%

\begin{lemma} \label{lemma_*-stab}
Let, for a B(a)-cluster $Cl_x^\emph{\textbf{T}}(R,C)$,
$\check{r}<\hat{r}$ and $(\sharp)$ hold, and the dominating part
restrictions in the cluster be $\emph{\textbf{T}}$-satisfied. Then
any generating rule can always be applied in a way to yield
$\emph{\textbf{T}}'$ such that the cluster to be
$\emph{\textbf{T}}'$-satisfied.
\end{lemma}

Lemma \ref{lemma_*-stab} shows that $(\sharp)$ also guarantees the
stability of the non-deterministic c.t.-sa\-tis\-fiability, namely
stable $\big(\big\lfloor{1\over {\hat{r}-\check{r}}}\big\rfloor
+1\big)$-satisfiability. It is clear that being once fulfilled,
$(\sharp)$ holds for any greater number of $R$-neighbours of $x$,
and so c.t.-sa\-tis\-fying of the dominating part restrictions can
be preserved as $R^\textbf{T}(x)$ grows.

Thus, Lemma\thinspace$\ref{lemma_*-sat}$ and
Lemma\thinspace$\ref{lemma_*-stab}$ guarantee for a
c.t.-satisfiable (with $\check{r}<\hat{r}$) B(a)-cluster
$Cl_x^\textrm{\textbf{T}}(R,C)$ that, having the number of
$R$-neighbours of $x$ equal to, or greater than
$\big\lfloor{1\over {\hat{r}-\check{r}}}\big\rfloor +1$ (what we
will call the \emph{border amount of neighbours} of $x,\ BAN_x$,
for that cluster), the cluster can be non-deterministically
c.t.-satisfied. Then, the termination of application of rules,
triggered by (the part restrictions in) that cluster, is ensured
by the check-up for $|R^\textrm{\textbf{T}}(x)|$.

Shortly said, any c.t.-satisfiable B(a)-cluster can be
non-deterministically stably c.t.-satisfied when the node has
enough many neighbours on the role in the cluster. We will rate
that in the general case for all (possibly counteracting) part
restrictions, to preserve from infinite application of part
rules.\medskip






\noindent \textbf{Clusters of type B(b).} B(b)-clusters are
determined by the equality $4^\circ$(b) $r_3+r_4=1$ for the
indices in $W$ part restrictions. These clusters are
c.t.-satisfiable if $2^\circ$ $r_1<r_4$ and $3^\circ$ $r_2<r_3$
hold (in case that the corresponding $M$ part restrictions are in
the cluster; in that case $1^\circ$ obviously also holds). Thus,
if $2^\circ$ and $3^\circ$ hold, or some $M$ part restriction is
missing, $4^\circ$(b) can be considered as a sufficient condition
for the c.t.-satisfiability of a B(b)-cluster.

%


\begin{lemma} \label{lemma_Bb}
Let, for a B(b)-cluster $Cl_x^\emph{\textbf{T}}(R,C)$, $r_1<r_4$
and $r_2<r_3$ hold, in case the corresponding $M$-part
restrictions are in the cluster. Then the cluster is
c.t.-satisfiable, and the sufficient condition it to be
non-deterministically c.t.-satisfied is the number of
$R$-neighbours of $x$ to be devisable by the denominator
of $r_3$ and $r_4$ from the cluster.
\end{lemma}




\subsubsection{The general case}


Let us recall that the part rules require all possible
applications of non-generating rules for the current node to be
already done, what ensures all possible (for the moment) concepts,
including part restrictions, to be already present in the node's
label. Generalizing the considerations for counteracting in
clusters, also taking into account the other concepts, triggering
generating rules, and using the indices technique, we prove:

\begin{lemma} \label{lemma:BAN}
Let $x$ be a node of a completion tree \emph{\textbf{T}}, and let
all possible applications of non-generating rules for $x$ be done.
Then it can be calculated a natural number $BAN_x\geq 1$,
depending on the concepts in $\mathcal{L}(x)$, and if $x$ is a
successor of $u$, possibly depending also on the concepts in
$\mathcal{L}(u)$, and having the following property: all part
restrictions in $\mathcal{L}(x)$ which are simultaneously
\emph{\textbf{T}}-satisfiable can be non-deterministically
simultaneously \emph{\textbf{T}}-satisfied when the number of
neighbours of $x$ on any role at the uppermost level in these part
restrictions becomes equal to $BAN_x$.
\end{lemma}

Lemma \ref{lemma:BAN} both legitimates the use of $BAN_x$ in the
part rules applicabi\-lity check-up, thus ensuring termination,
and guarantees that all simultaneously c.t.-satisfiable part
restrictions will be non-deterministically c.t.-satisfied, so that
there would not be clashes with them in the complete c.t.

Note that the border amount of neighbours can change only if
$\mathcal{L}(x)$, or $\mathcal{L}(u)$ be changed, for example by
adding of some concept to any of them caused by an application of
a rule for a successor. As the number of such possible changes is
limited by the number of concepts in $clos(D)$, after finite
number of recalculations we will obtain the final for the node $x\
BAN_x$.

\section{Correctness of the Algorithm}

As usual with tableaux algorithms we prove lemmas that the
algorithm always terminates, and that it is sound and complete.
The termination is ensured by pair-wise blocking, and by
$BAN$-checkup, which guarantees finite (at most exponential---in
case of the usual binary coding of numbers) branching at a node.
The build of a tableau from the completion tree and the reverse
follows the constructions from \cite{HST99}, Section 5.4, Lemmas
16 and 17. Since the \emph{internalization of terminologies}
\cite{Baa90} is still possible in the presence of part
restrictions, following the technique presented in \cite{HST99},
Section 3.1, we obtain finally:

\begin{theorem} \label{main-DP}
The presented tableaux algorithm is a decision procedure for the
satisfiability and subsumption of
$\mathcal{ALCQPIH}_{R^+}$-concepts with respect to role
hierarchies and terminologies.
\end{theorem}

\section{Conclusion}

DL $\mathcal{ALCQPIH}_{R^+}$ augments $\mathcal{ALCQIH}_{R^+}$
with the ability to express rational grading. We showed that the
decision procedure for the latter logic can be naturally extended
to capture the new one. This indicates that the approach which
realizes rational grading independently from integer grading is
fruitful, and can be applied even to expressive description logics
to give in a convenient way their rational grading extensions,
still keeping the decidability.\bigskip

\noindent \textbf{Acknowledgments}\smallskip

\noindent We would like to thank the anonymous reviewers for their
valuable questions, remarks and suggestions.


\newpage
\label{sect:bib}
\bibliographystyle{plain}
\bibliography{Yanchev_ref}

\newpage
\section*{Appendix}

In this appendix we present the proofs of Lemmas 1--6 and Theorem
1---complete, only for the new cases, or just sketched. For better
readability we also give the complete definition of the
interpretation function and the complete set of the tableaux
properties, the set of clashes and the complete set of completion
rules, and repeat the propositions.
\bigskip

For an \emph{interpretation} $\mathcal{I} = (\Delta^\mathcal{I},
\cdot^\mathcal{I})$, for any concepts $C,\ D$, role $R$, $n\geq
0$, and rational number $r\in \textbf{Q}_0$, the complete
inductive definition of the interpretation function
$\cdot^\mathcal{I}$ follows. $R^\mathcal{I}(x)$ denotes the set of
objects which are $R^\mathcal{I}$-related to
($R^\mathcal{I}$-neighbours of) the object $x$, $R^\mathcal{I}(x,
C)$ denotes the set of $R^\mathcal{I}$-neighbours of $x$ which are
in $ C^\mathcal{I}$, i.e., the set $\{y\ |\ \langle x,y\rangle\in
R^\mathcal{I}$ and $y\in C^\mathcal{I}\}$, and $\sharp M$ denotes
the cardinality of a set $M$.
\begin{align*}
A^\mathcal{I} &\subseteq \Delta^\mathcal{I} \textrm{ for any } A\in\textbf{C}_0 \\
(\neg C)^\mathcal{I} &= \Delta^\mathcal{I}\backslash C^\mathcal{I} \\
( C\sqcap D)^\mathcal{I} &= C^\mathcal{I}\cap D^\mathcal{I}\\
( C\sqcup D)^\mathcal{I} &= C^\mathcal{I}\cup D^\mathcal{I}\\
(\exists R. C)^\mathcal{I} &= \{x\in \Delta^\mathcal{I}\ |
\textrm{ there is some } y \in\Delta^\mathcal{I} \textrm{ with }
\langle x,y\rangle \in R^\mathcal{I} \textrm{ and } y
\in C^\mathcal{I}\}\\
(\forall R. C)^\mathcal{I} &= \{x\in \Delta^\mathcal{I}\ |
\textrm{ for all } y\in\Delta^\mathcal{I}, \textrm{ if } \langle
x,y\rangle \in R^\mathcal{I}\textrm{ then } y \in C^\mathcal{I}\}\\
(\leqslant nR. C)^\mathcal{I} &= \{x\in \Delta^\mathcal{I}\ |\
\sharp R^\mathcal{I}(x, C) \leq n \} \\
(\geqslant nR. C)^\mathcal{I} &= \{x\in \Delta^\mathcal{I}\ |\
\sharp R^\mathcal{I}(x, C) \geq n \}\\
(MrR. C)^\mathcal{I} &= \{ x\in \Delta^\mathcal{I}\ |\
\sharp R^\mathcal{I}(x, C) > r.\sharp R^\mathcal{I}(x)\} \\
(WrR. C)^\mathcal{I} &= \{ x\in \Delta^\mathcal{I}\ |\ \sharp
R^\mathcal{I}(x,\neg C) \leq r.\sharp R^\mathcal{I}(x)\}
\end{align*}

Also, for any $S\in \textbf{R}$ and any $R\in \textbf{R}^+$ we
define:
\begin{align*}
\langle x,y\rangle\in S^\mathcal{I}\ & \textrm{iff}\ \langle
y,x\rangle\in (\textsf{Inv}(S))^\mathcal{I}\\
\textrm{if}\ \langle x,y\rangle\in R^\mathcal{I}\ \textrm{and}\
\langle y,z\rangle\in R^\mathcal{I}\ & \textrm{then}\ \langle
x,z\rangle\in R^\mathcal{I}
\end{align*}

\bigskip


A \emph{tableau} $T=(\textbf{S}, \mathcal{L}, \mathcal{E})$ for a
concept $D$ in NNF with respect to a role hierarchy
$\mathcal{R}^+$ must satisfy, for all individuals $s,t\in
\textbf{S}$, concepts $C,C_1,C_2 \in clos(D)$, and roles $R,S\in
\emph{\textbf{R}}_D$, the following 13 properties.

We denote with $R^T(s)$ the set of individuals, $R$-related to
$s$, and with  $R^T(s, C)$ ---the set of individuals, $R$-related
to $s$ and labelled with $C$, i.e., $R^T(s, C):=\{t\in
{\textbf{S}}\ |\ \langle s,t\rangle\in \mathcal{E}(R)\
\textrm{and}\ C \in \mathcal{L}(t)\}$.

$\boxtimes$ is a placeholder for $\geqslant n$ and $\leqslant n$,
for arbitrary $n\geq 0$, for $\exists$, and for $Mr$ and $Wr$, for
arbitrary $r\in \textbf{Q}_0$.\medskip

$\phantom{0}1.\ \textrm{If}\ C\in \mathcal{L}(s),\ \textrm{then}\
\neg C \not\in \mathcal{L}(s).$\medskip

$\phantom{0}2.\ \textrm{If}\ C_1\sqcap C_2 \in \mathcal{L}(s),\
\textrm{then}\ C_1\in \mathcal{L}(s) \textrm{ and } C_2 \in
\mathcal{L}(s).$\medskip

$\phantom{0}3.\ \textrm{If}\ C_1\sqcup C_2 \in \mathcal{L}(s),\
\textrm{then}\ C_1\in \mathcal{L}(s) \textrm{ or } C_2 \in
\mathcal{L}(s).$\medskip

$\phantom{0}4.\ \textrm{If}\ \forall R.C\in \mathcal{L}(s)
\textrm{ and } \langle s,t \rangle \in\mathcal{E}(R),\
\textrm{then}\ C\in\mathcal{L}(t).$\medskip

$\phantom{0}5.\ \textrm{If}\ \exists R.C\in \mathcal{L}(s),\
\textrm{then there is some}\ t\in\textbf{S} \textrm{ such that }
\langle s,t \rangle \in\mathcal{E}(R) \textrm{ and } \\
\indent\phantom{05.}\ C \in\mathcal{L}(t).$\medskip

$\phantom{0}6.\ \textrm{If}\ \forall S.C\in \mathcal{L}(s)
\textrm{ and } \langle s,t \rangle \in\mathcal{E}(R) \textrm{ for
some } R\sqsubseteq^+ S \textrm{ with } \textsf{Trans}(R),\ \textrm{then}\ \\
\indent\phantom{06.}\ \forall R.C\in\mathcal{L}(t).$\medskip

$\phantom{0}7.\ \langle s,t \rangle \in\mathcal{E}(R) \textrm{ iff
}  \langle t,s \rangle \in\mathcal{E}\big ( \textsf{Inv}(R)\big
).$\medskip

$\phantom{0}8.\ \textrm{If}\ \langle s,t \rangle \in\mathcal{E}(R)
 \textrm{ and } R\sqsubseteq^+ S, \textrm{ then } \langle s,t \rangle
 \in\mathcal{E}(S).$\medskip

$\phantom{0}9.\ \textrm{If}\ \leqslant nR.C\in \mathcal{L}(s),\
\textrm{then}\ \sharp R^T(s,C)\leq n.$\medskip

$10.\ \textrm{If}\ \geqslant nR.C\in \mathcal{L}(s),\
\textrm{then}\ \sharp R^T(s,C)\geq n.$\medskip

$11.\ \textrm{If}\ MrR.C\in \mathcal{L}(s),\ \textrm{then}\ \sharp
R^T(s,C)>r.\sharp R^T(s).$\medskip

$12.\ \textrm{If}\ WrR.C\in \mathcal{L}(s),\ \textrm{then}\ \sharp
R^T(s,\sim\! C)\leq r.\sharp R^T(s).$\medskip

$13.\ \textrm{If}\ \boxtimes R.C\in \mathcal{L}(s)\ \textrm{and}\
\langle s,t\rangle \in \ \mathcal{E}(R),\ \textrm{then}\ C\in
\mathcal{L}(t)\textrm{ or } \sim\! C\in\mathcal{L}(t).$

\bigskip


\noindent \textbf{Lemma~1.} \emph{An
$\mathcal{ALCQPIH}_{R^+}$-concept D is satisfiable with respect to
a role hierarchy $\mathcal{R}^+$ iff there exists a tableau for D
with respect to $\mathcal{R}^+$.}

\begin{proof}
We extend the proof of Lemma~14 in \cite{HST99} for the modified
and the new cases in both the interpretation and the tableau.

For the \emph{if} direction, let $T=({\textbf{S}}, \mathcal{L},
\mathcal{E})$ be a tableau for the tested concept $D$ with $D\in
\mathcal{L}(s_0)$. We define a model $\mathcal{I} =
(\Delta^\mathcal{I}, \cdot^\mathcal{I})$ as in \cite{HST99}:
\begin{eqnarray*}
  \Delta^\mathcal{I} &=& \textbf{S} \\
  A^\mathcal{I} &=& \{s\ |\ A\in\mathcal{L}(s)\},\textrm{ for all }A
  \in\textbf{C}_0\cap clos(D) \medskip \\
  R^\mathcal{I} &=& \Biggl \{
  \begin{array}{ll}
    \mathcal{E}(R)^+, & \textrm{ if }R\textrm{ is transitive} \smallskip \\
    \mathcal{E}(R)\cup
    \underset{P\sqsubseteq^+R,P\not =R}\bigcup  P^\mathcal{I}, & \textrm{ otherwise}
  \end{array}
\end{eqnarray*}

In \cite{HST99}, Lemma~14 it is noted that
$\mathcal{I}\models\mathcal{R}^+$ due to Property 8 of tableaux
definition. To complete the proof that $\mathcal{I}$ is a model of
$D$ w.r.t. $\mathcal{R}^+$, we have to show that from
$C\in\mathcal{L}(s)$ it follows that $s\in C^\mathcal{I}$ for any
$s\in \textbf{S}$, in case that $C$ is a part restriction, what
will yield $D^\mathcal{I}\not = \emptyset$. We extend the
inductive definition of the norm $\|C\|$ of a concept $C$ in NNF
to cover the part restrictions:
\begin{eqnarray*}
\begin{array}{lll}
  \|A\| & :=\ \|\neg A\| & :=\  0\textrm{ for }A\in \textbf{C}_0\\
  \|C_1\sqcap C_2\| & :=\  \|C_1\sqcup C_2\| & :=\  1+\|C_1\|+\|C_2\| \\
  \|\forall R.C\| & :=\  \|\exists R.C\| & :=\  1+\|C\| \\
  \|\geqslant nR.C\| & :=\  \|\leqslant nR.C\| & :=\  1+\|C\| \\
  \|MrR.C\| & :=\  \|WrR.C\| & :=\  1+\|C\| \\
\end{array}
\end{eqnarray*}

We consider the inductive steps with part restrictions. Recall the
limitation that roles in part restrictions are simple, so for a
role $R$ in a part restriction it holds
$R^\mathcal{I}=\mathcal{E}(R)$.

\begin{itemize}
    \item [$\bullet$] $C=MrR.E$.
    We have to show for an $s$ with $MrR.E\in \mathcal{L}(s)$
    that $s\in (MrR.E)^\mathcal{I}$, or that
    $\sharp R^\mathcal{I}(s,E)>r.\sharp R^\mathcal{I}(s)$.
    From $MrR.E\in \mathcal{L}(s)$ we have that $\sharp R^T(s,E)>r.\sharp
    R^T(s)$. As $R$ is simple, $t\in R^\mathcal{I}(s)$
    iff $\langle s,t \rangle\in
    R^\mathcal{I}$ iff $\langle s,t \rangle\in
    \mathcal{E}(R)$ iff $t\in R^T(s)$, so $\sharp R^\mathcal{I}(s)=\sharp R^T(s)$.
    From the inductive hypothesis for $E$, from $E\in
    \mathcal{L}(t)$ it follows that $t\in E^\mathcal{I}$,
    so $\sharp R^\mathcal{I}(s,E)\geq \sharp R^T(s,E)$. Thus, we obtain
    $\sharp R^\mathcal{I}(s,E) \geq \sharp R^T(s,E) > r.\sharp R^T(s)$ =
    $r.\sharp R^\mathcal{I}(s)$. \medskip

    \item [$\bullet$] $C=WrR.E$.
    We have to show for an $s$ with $WrR.E\in \mathcal{L}(s)$
    that $s\in (WrR.E)^\mathcal{I}$, or that
    $\sharp R^\mathcal{I}(s,\sim\!\! E)\leq r.\sharp R^\mathcal{I}(s)$.
    From $WrR.E\in \mathcal{L}(s)$ we have that $\sharp R^T(s,\sim\!\! E)\leq r.\sharp
    R^T(s)$. By the definition, $t\in
    R^T(s)$ iff $\langle s,t \rangle\in\mathcal{E}(R)$ iff $\langle s,t \rangle\in
    R^\mathcal{I}$ (as $R$ is simple, $R^\mathcal{I}=\mathcal{E}(R)$) iff
    $t\in R^\mathcal{I}(s)$. So $R^T(s)=R^\mathcal{I}(s)$. Let
    assume that $\sharp R^T(s,\sim\!\! E) < \sharp R^\mathcal{I}(s,\sim\!\!
    E)$. Then there is some $t\in R^\mathcal{I}(s) = R^T(s)$ with
    $t\in (\sim\!\! E)^\mathcal{I}=\Delta^\mathcal{I}\backslash
    E^\mathcal{I}$ and $\sim\!\! E\not\in\mathcal{L}(t)$. From
    property 13 of the tableaux definition it follows that $E\in\mathcal{L}(t)$,
    so, from the inductive hypothesis for $E$, $t\in E^\mathcal{I}$,
    what is a contradiction with $t\in \Delta^\mathcal{I}\backslash E^\mathcal{I}$.
    So we have $\sharp R^\mathcal{I}(s,\sim\!\! E) \leq \sharp R^T(s,\sim\!\! E)
    \leq r.\sharp R^T(s) = r.\sharp R^\mathcal{I}(s)$.
\end{itemize}

For the \emph{only if} direction, for a satisfiable concept $D$
with a model $\mathcal{I} = (\Delta^\mathcal{I},
\cdot^\mathcal{I})$, $\mathcal{I}\models\mathcal{R}^+$, we again
define a tableau as in \cite{HST99}, Lemma 14:
\begin{eqnarray*}
  \textbf{S} &=& \Delta^\mathcal{I} \\
  \mathcal{L}(s) &=& \{C\in clos(D)\ |\ s\in C^\mathcal{I}\} \\
  \mathcal{E}(R) &=& R^\mathcal{I}
\end{eqnarray*}

Then, the new properties 11. and 12. are satisfied due to the
corresponding definitions in the semantics, and the modified
property 13. is also satisfied as a consequence of the definition
of semantics of concepts. $\circlearrowleft$\bigskip
\end{proof}

\noindent \textbf{Lemma~2.} \emph{For a B(a)-cluster
$Cl_x^\emph{\textbf{T}}(R,C)$,} \emph{the inequalities $1^\circ$,
$2^\circ$, $3^\circ$, and $4^\circ$\emph{(a)}, with the
corresponding part restrictions being in the cluster, hold iff}
\[
\check{r}<\hat{r}
\]

\begin{proof}
A straightforward inspection. $\circlearrowleft$\bigskip
\end{proof}

\noindent \textbf{Lemma~3.} \emph{Let, for a B(a)-cluster
$Cl_x^\emph{\textbf{T}}(R,C)$, $\check{r}<\hat{r}$ hold. Then}
\[
|R^\textbf{T}(x)| > {1\over {\hat{r}-\check{r}}}  \quad (\sharp)
\]
\noindent \emph{is a sufficient condition for the
non-deterministic $|R^\emph{\textbf{T}}(x)|$-satisfiability of the
cluster.}
\begin{proof}

Let $|R^\textbf{T}(x)|=n > 0$, and  $(\sharp)$ holds. We can take
non-deterministically an ordering of $R^\textbf{T}(x):
R^\textbf{T}(x)= \{y_1,...,y_n\}$. Let then
\[\big\lfloor \check{r}.|R^\textbf{T}(x)|\big\rfloor =i\]
\[\big\lceil \hat{r}.|R^\textbf{T}(x)|\big\rceil\break =j\]

Due to $(\sharp)$, $j-i\geq 2$, and so $k$ exists such that $0\leq
i<k<j\leq n$. Then, labelling with $C$ $y_1,...,y_k$, and with
$\sim\!\!C$---$y_{k+1},...,y_n$ $\textbf{T}$-satisfies the
dominating part restrictions. It is straightforward check-up to
see this in different cases for $\check{r}$ and $\hat{r}$.


Obviously, the same result can be obtained in at least one
non-deterministic labelling of $x$-successors with $(\sim)C$,
without using any ordering of $R^\textbf{T}(x)$.
$\circlearrowleft$\bigskip
\end{proof}

\noindent\textbf{Lemma~4.} \emph{Let, for a B(a)-cluster
$Cl_x^\emph{\textbf{T}}(R,C)$, $\check{r}<\hat{r}$ and $(\sharp)$
hold, and the dominating part restrictions in the cluster be
$\emph{\textbf{T}}$-satisfied. Then any generating rule can always
be applied in a way to yield a completion tree
$\emph{\textbf{T}}'$ such that the cluster is
$\emph{\textbf{T}}'$-satisfied.}

\begin{proof}

Let $E\in\{ C,\sim\!C\}$. We will show that if we add a new
$R$-successor of $x$, and its labelling with $E$ makes one of the
dominating part restrictions c.t.-unsatisfied, then the labelling
with $\sim\!E$ keeps both of them c.t.-satisfied.

We present the proof in one of the possible cases for $\check{r}$,
$\hat{r}$, and $ E $. The whole proof consecutively shows, in just
the same manner, that the lemma holds in  all other cases.

First note that if the dominating part restrictions in a
B(a)-cluster are both with $M$ constructor, then
$\check{r}<\hat{r}$ is just the inequality $1^\circ$, and
${\hat{r}-\check{r}}$ takes the form $1-r_{M1}-r_{M2}$; if they
are one with $M$ and one with $W$ constructor---we have either
$2^\circ$, or $3^\circ$, and ${\hat{r}-\check{r}}$ takes the form
$r_{W}-r_{M}$; if they are both with $W$---$4^\circ$(a) and
${\hat{r}-\check{r}}$ takes the form $r_{W1}+r_{W2}-1$. \smallskip

Now let $\check{r}=r_1$, $\hat{r}=1-r_2$, and $E =\thinspace
\sim\!C$. Then $\check{r}<\hat{r}$ and ($\sharp$) have the forms:
\[
 r_1+r_2<1 \quad (\check{r}<\hat{r})\text{, and}
\]
\[
|R^\textbf{T}(x)|-r_1.|R^\textbf{T}(x)|-{r}_2.|R^\textbf{T}(x)|> 1
\quad (\sharp)
\]

From the $\textbf{T}$-satisfaction of the dominating part
restrictions we have also:

\[
|R^\textbf{T}(x, C)|>r_1.|R^\textbf{T}(x)|,\ 
\]
\[
|R^\textbf{T}(x,\sim\!C)|>r_2.|R^\textbf{T}(x)|
\]

Let the labelling with $\sim\!C$ of the new variable $y$ (yielding
the c.t. $\textbf{T}'$) make $M{r_1}R.C$ not
$\textbf{T}'$-satisfied, i.e., \begin{equation}
\label{eq_lemma-stab} |R^\textbf{T}(x, C)| = |R^{\textbf{T}'}(x,
C)| \leq r_1.|R^{\textbf{T}'}(x)| = r_1.(|R^\textbf{T}(x)|+1) =
r_1.|R^\textbf{T}(x)|+r_1
\end{equation}

So, let the generating rule labels $y$ with $C$ instead (yielding
the c.t. $\textbf{T}''$). Then the part restriction with $M{r_1}$
is obviously $\textbf{T}''$-satisfied:
\[
|R^{\textbf{T}''}(x, C)| = |R^\textbf{T}(x, C)|+1 >
r_1.|R^\textbf{T}(x)|+1
>
\]
\[
>r_1.|R^\textbf{T}(x)|+r_1 = r_1.(|R^\textbf{T}(x)|+1) = r_1.|R^{\textbf{T}''}(x)|
\]

To complete this case of the proof we have to show that the part
restriction with $M{r_2}$ is also $\textbf{T}''$-satisfied.
Indeed,

\[
|R^{\textbf{T}''}(x,\sim\!C)| = |R^\textbf{T}(x,\sim\!C)| =
|R^\textbf{T}(x)|-|R^\textbf{T}(x,
 C)| \overset{\textrm{from (\ref{eq_lemma-stab})}}\geq
\]
\[
\geq |R^\textbf{T}(x)|-r_1.|R^\textbf{T}(x)|-r_1
\overset{\textrm{from } (\sharp)}> 1+ r_2.|R^\textbf{T}(x)|-r_1
\overset{\textrm{from } \check{r}<\hat{r}}>
\]
\[
 > r_2.|R^\textbf{T}(x)|+r_2 =
r_2.(|R^\textbf{T}(x)|+1) = r_2.|R^{\textbf{T}''}(x)|
\circlearrowleft
\]
\end{proof}


Note that if one of the dominating part restrictions is with $W$
constructor, in both Lemma 3 and Lemma 4 $(\sharp)$ can be not
strict, but we take the stronger inequality to cover all cases.
\bigskip

\noindent \textbf{Lemma~5.} \emph{Let, for a B(b)-cluster
$Cl_x^\emph{\textbf{T}}(R,C)$, $r_1<r_4$ and $r_2<r_3$ hold, in
case the corresponding $M$ part restrictions are in the cluster.
Then the cluster is $\emph{\textbf{T}}$-satisfiable, and the
sufficient condition it to be non-deterministically
$\emph{\textbf{T}}$-satisfied is the number of $R$-neighbours of
$x$ to be devisable by the denominator
of $r_3$ and $r_4$ from the cluster.}

\begin{proof} (Idea) The condition for the non-deterministic
c.t.-satisfying of a B(b)-cluster can be easily seen if we look at
the simplest case in which such a cluster has part restrictions
with $r_3=r_4={1\over 2}$. Obviously, the condition both part
restrictions to be non-deterministically simultaneously
c.t.-satisfied is $x$ to have an even number of neighbours (if
there are at all) on the role in the cluster. $\circlearrowleft$
\bigskip
\end{proof}

\noindent \textbf{Lemma~6.} \emph{Let $x$ be a node of a
completion tree \emph{\textbf{T}}, and let all possible
applications of non-generating rules for $x$ be done. Then it can
be calculated a natural number $BAN_x\geq 1$, depending on the
concepts in $\mathcal{L}(x)$, and if $x$ is a successor of $u$,
possibly depending also on the concepts in $\mathcal{L}(u)$, and
having the following property: all part restrictions in
$\mathcal{L}(x)$ which are simultaneously
\emph{\textbf{T}}-satisfiable can be non-deterministically
simultaneously \emph{\textbf{T}}-satisfied when the number of
neighbours of $x$ on any role at the uppermost level in these part
restrictions becomes equal to $BAN_x$.}

\begin{proof}
The proof shows how the number $BAN_x$ with the above property is
calculated. Depending on the language there may be some
distinctions and features, and we will point them for the
considered DL. As a whole the calculation uses the main schema of
the indices technique.
\smallskip

First, we make some simplification. Any sub-concept of the form
$\exists R.C$ can be replaced with $\geqslant\! 1R.C$.
It is straightforward to show that the resulting concept is
satisfiable iff the initial one is. Also, the replacement takes
linear space and time in the length of the initial concept. So, we
assume in the following considerations without loss of generality
that the concept to be tested for satisfiability does not contain
sub-concepts of the form $\exists R.C$.
Thus, the concepts in $\mathcal{L}(x)$ triggering the generating
rules can be `at least' qualifying number restrictions or part
restrictions. We investigate the number of neighbours needed to
simultaneously c.t.-satisfy all part restrictions, taking into
account the (possible) presence of `at least' qualifying number
restrictions, c.t.-satisfied by $\geqslant$-rule.\smallskip

Next, we take care for the presence in the language of inverse
roles. If $x$ is an $\textsf{Inv}(S)$-successor of $u$, $u$ is an
$S$-neighbour of $x$. To ensure the correct calculations, we
assume that for each concept $E\in\mathcal{L}(u)$ we have in
$\mathcal{L}(x)$ the concept $\geqslant 1S.E$ (which would be,
obviously, already c.t.-satisfied, if present in
$\mathcal{L}(x)$).
\smallskip

Now we calculate $BAN_x$ looking only at $n$-s in `at least'
qualifying number restrictions, and $r$-s in part restrictions, to
guarantee the desired property. We take into account that `at
least' qualifying number restrictions can affect the
c.t.-satisfying of part restrictions. The application of
$\geqslant$-rule for $\geqslant\! nS.E$ results in $n\
S$-neighbours of $x$ labelled with $E$, and the c.t.-satisfaction
of $MrR.C$ will depend on that, in case $S\sqsubseteq^+ R$, and
$C$ and $E$ contradict each other (i.e., labelling with them the
same neighbour of $x$ leads to a clash). To deal with that, we use
the indices technique. Using the above notation, we want
$n<(1-r).|R^\textbf{T}(x)|$. If we have $W$ instead of $M$
constructor, the inequality is $n\leq r.|R^\textbf{T}(x)|$. It is
clear now that a sufficient condition for the (possible)
contradiction of $C$ with $E$ to be overcome, so that the part
restriction to be $\textbf{T}$-satisfied, is
$|R^\textbf{T}(x)|\geq\big\lfloor{n \over {1-r}}\big\rfloor+1$,
for $r$-s in $M$ part restrictions, and
$|R^\textbf{T}(x)|\geq\big\lfloor{n \over {r}}\big\rfloor+1$, for
$r$-s in $W$ part restrictions. Also, as all `at least'
restrictions in $\mathcal{L}(x)$ can affect some part restriction,
we consider:
\[
\bar n = \sum n,
\]
\noindent where the sum is on all $n$-s in all `at least'
restrictions in $\mathcal{L}(x)$, together with those assumed
instead of $\exists$-concepts, and due to $\mathcal{L}(u)$, as
described above. Due to the assumption concerning
$\mathcal{L}(u)$, $\bar n \geq 1$.
\medskip



Another feature that we have to consider is the role hierarchy. In
the following considerations we have in mind the fixed role
hierarchy, omitting to point it explicitly. For the correct
dealing with part restrictions we need to clearly distinguish the
roles, so we assume w.l.o.g. that if two roles have different
notations they are really different. In particular, if
$S\sqsubseteq^+ R$, we assume that $S$ is a real sub-role of $R$.

Now, we turn again to the indices technique. We know from the
considerations there that the satisfaction of all inequalities
$1^\circ$--$4^\circ$ by the indices is a necessary condition for
part restrictions' simultaneous c.t.-satisfaction, and that
failing some of the inequalities leads to an unavoidable clash (or
to infinite tree generation, if we do not stop the process). But
the latter is true only if the part restrictions are with one and
same role. If the roles are independent, the c.t.-satisfaction is
guarantied. A special case is when there is an inclusion of roles.
Let consider, for example, the concept $M{1\over 2}S.C\sqcap
M{1\over 2}R.(\sim\! C)$ labelling a node $x$, where
$S\sqsubseteq^+ R$. Nevertheless the $r$-s failed the
corresponding inequality ($1^\circ$), the concept can be easily
c.t.-satisfied, if $x$ has one $S$-neighbour, labelled with $C$
(c.t.-satisfying the first sub-concept), and two more
$R$-neighbours, labelled with $\sim\! C$ (all three neighbours
c.t.-satisfying the second sub-concept), as the $R$-neighbours
does not affect the c.t.-satisfaction of the first sub-concept. In
such a case the (number of) neighbours, added so c.t.-satisfy the
(sub-)concept with the sub-role serve as starting point for
c.t.-satisfying the one with the sup-role. The situation is
similar with the one with $\bar{n}$ neighbours necessary for the
`at least' sub-concepts.\smallskip

Let denote the set of roles at the uppermost level in part
restrictions in $\mathcal{L}(x)$ with $\textbf{R}^P_x$. All roles
in $\textbf{R}^P_x$ can be divided in non-intersecting chains with
respect to $\sqsubseteq^+$. Fully independent roles form (trivial)
chains by themselves, and there can be linearly ordered
($R_1\sqsubseteq^+ R_2\sqsubseteq^+ R_3$) and/or partially ordered
(e.g. $R_1\sqsubseteq^+ R_2$, $R_1\sqsubseteq^+ R_3$,
$R_1\sqsubseteq^+ R_4$, $R_2\sqsubseteq^+ S$, and
$R_3\sqsubseteq^+ S$) chains.

For two roles $R$ and $S$ from $\textbf{R}^P_x$, we call $R$ a
\emph{direct sub-role} of $S$ \emph{in} $\textbf{R}^P_x$ if
$R\sqsubseteq^+ S$ and there is no such $R_1 \in \textbf{R}^P_x$
that $R\sqsubseteq^+ R_1\sqsubseteq^+ S$. In this case $S$ is a
\emph{direct sup-role} of $R$ (in $\textbf{R}^P_x$). We call a
role in $\textbf{R}^P_x$ \emph{initial} if it has no sub-role in
$\textbf{R}^P_x$, and \emph{final} if it is not a sub-role of any
other role in $\textbf{R}^P_x$. Note that all independent roles
are both initial and final ones.
\smallskip

We first enumerate arbitrary the chains, and then enumerate the
roles in any chain, starting with the initial ones, and keeping
any sup-role to obtain a number, greater than the one of any of
its sub-roles. For the $j$-th role in the $i$-th chain ($i,j$-role
for short) we call the \emph{base number}, and denote it $B^i_j$,
the number of neighbours which are present when we start to
c.t.-satisfy the (sub-)concepts with that role. Also, for
$i,j$-role, we call the \emph{satisfaction number}, and denote it
$N^i_j$, the number of neighbours sufficient for all part
restrictions with that role to be simultaneously c.t.-satisfied.
Now, it is clear that $\bar{n}$ is the base number for all initial
roles, and for a non-initial role the base number is the sum of
satisfaction numbers of all its direct sub-roles. Thus, in the
partial ordered example above (if that is the only chain, or the
one with number $1$) the satisfaction number of $R_2$ is denoted
$N^1_2$, the base number of $S$ is then $B^1_5=N^1_2+N^1_3$, and
its satisfaction number is $N^1_5$.
\medskip

Now we will see how to calculate the number of neighbours,
sufficient for the simultaneous c.t.-satisfaction of all part
restrictions with the same role (in case their indices satisfy the
inequalities $1^\circ$--$4^\circ$). In the presence of complex
concepts the counteraction can be not directly seen, and a cluster
for a role $R$ can contain more than two (dominating)
counteracting concepts (and the domination is also not clear). We
denote with $Cl^i_j$ the cluster for the $i,j$-role (omitting in
the notation the tree $\textbf{T}$ and the node $x$), with all
part restrictions in $\mathcal{L}(x)$ with $i,j$-role at the
uppermost level, and let $|Cl^i_j|=p^i_j$.\smallskip

We assume that all part restrictions in $Cl^i_j$ can counteract.
To simplify considerations (and also to keep them correct, as if
we consider more than two part restrictions, the sum of their
$r$-s can be bigger not only than $1$, but also than $2$, etc.) we
reduce the four inequalities $1^\circ$--$4^\circ(a)$ just to
$1^\circ$ using the semantics of $W$. $WrR.C$ is (c.t.-)satisfied
in $x$ if $\sim\!\!C$ labels not greater than $r$ part of
$R$-neighbours of $x$, i.e., if $C$ labels not less than $1-r$
part of $R$-neighbours of $x$. But as we are looking for a
sufficient number of neighbours (not necessary the minimal one),
we can replace `not less than' with `greater than', which gives
the semantics of $M$ constructor. So we can replace each $r$ in
$\mathcal{L}(x)$ with $\bar{r}$, defined as follows:
\begin{eqnarray*}
  \bar{r} &=& \Big \{
  \begin{array}{ll}
    r & \textrm{, if r is in a part restriction with }M
    \textrm{ constructor} \smallskip\\
    1-r & \textrm{, if r is in a part restriction with }W
    \textrm{ constructor}
  \end{array}
\end{eqnarray*}

$\bar{r}$-s unify the shape of the sufficient conditions for
c.t.-satisfiability of concepts with $M$ and $W$ constructors, and
inequalities $2^\circ$, $3^\circ$ and $4^\circ(a)$ take the form
of $1^\circ$. Now the condition for overcoming the (counter)action
of all `at least' restrictions has the common shape
$|R^\textbf{T}(x)|\geq\big\lfloor{\bar{n} \over
{1-\bar{r}}}\big\rfloor+1$. \medskip

Also, as in the $Cl^i_j$ for $i,j$-role there are possibly up to
$p^i_j$ counteracting concepts, we have to consider all possible
$C^k_{p^i_j}$ combinations of concepts, $1\leq k\leq p^i_j$, with
indices, satisfying $1^\circ$. We include $k=1$ to capture the
case when $p^i_j=1$, i.e., when there is just one part restriction
in $Cl^i_j$, and there is no counteraction inside the cluster.
Thus, for any initial $i,j$-role we calculate its satisfaction
number to be:
\[
N^i_j=\max_{C^k_{p^i_j}} \Bigg(\bigg\lfloor{\bar{n} \over {1-{\sum
_{k_l} \bar{r}}}}\bigg\rfloor+k\Bigg),
\]

\noindent where the sum is on all $\bar{r}$-s in a fixed
combination $k_l$ with $k$ part restrictions, $1\leq l\leq
C^k_{p^i_j}$, such that $\sum _{k_l} \bar{r} < 1$, and the maximum
is taken on all possible such expressions for $1\leq k\leq p^i_j$
and $1\leq l\leq C^k_{p^i_j}$, if at least one such expression
exists. Note that the additive constant in the expression is $k$,
instead of $1$, as in the presence of $k$ counteracting $M$ part
restrictions we will need for their c.t.-satisfaction at least $k$
neighbours, no matter how small are the indices.

If there is no such combination $k_l$ with $k$ elements, $1\leq
k\leq p^i_j$, $1\leq l\leq C^k_{p^i_j}$, that $\sum _{k_l} \bar{r}
< 1$, we set $N^i_j=1$. \smallskip

Note that for initial roles we take $B^i_j=\bar n$, no matter on
which roles are `at least' restrictions. If we dare for
efficiency, we can reduce the base numbers of initial roles,
considering only `at least' restrictions on any fixed such role
and on its sub-roles.\medskip

In the common case for the non-initial $i,j$-role, we replace
$\bar{n}$ with $B^i_j$, as $\bar{n}$ is the base number only for
the initial roles:
\[
B^i_j = \sum_l N^i_l,
\]
\noindent where the sum is on the satisfaction numbers of all
direct sub-roles of $i,j$-role in the $i$-th chain.

For any non-initial $i,j$-role we calculate:
\[
N^i_j=\max_{C^k_{p^i_j}} \Bigg(\bigg\lfloor{B^i_j \over {1-{\sum
_{k_l} \bar{r}}}}\bigg\rfloor+k\Bigg),
\]
\noindent if there is at least one combination $k_l$ with $k$ part
restrictions, $1\leq k\leq p^i_j$, $1\leq l\leq C^k_{p^i_j}$, such
that $\sum _{k_l} \bar{r} < 1$, otherwise we set $N^i_j=1$.
\medskip

Thus, starting with the initial roles and taking $B^i_j=\bar n$ we
decide any possible counteractions between `at least' qualifying
number restrictions and part restrictions with initial roles. Next
we continue successively with the (direct) sup-roles taking for
any sup-role as its base number the sum of satisfaction numbers of
all its direct sub-roles. Thus we decide any possible
counteractions between part restrictions with the role, as well as
between part restrictions with the role from one side, and part
restrictions with all its sub-roles, from the other. Following
that way we end with the final roles in any chain.

When, in the formula for $N^i_j$, $k=p^i_j=1$, we have the case of
A-cluster (with no counteraction), and we guarantee
c.t.-satisfiability just taking into account the neighbours
already present (the base number). When $p^i_j > 1$ we have the
case of Ba-cluster, and we decide (anytime it is possible, with
$\sum _{k_l} \bar{r} < 1$) the counteractions in that cluster
(again, taking into account the neighbours already present). Thus,
having the number of neighbours on a final role equal to its
satisfaction number guarantees simultaneous c.t.-satisfiability of
all part restrictions with that role, which indices satisfy
inequalities $1^\circ$--$4^\circ(a)$, and of all part restrictions
(also satisfying the inequalities) with all sub-roles of that
role. So we define:\medskip

$BAN_{A\_Ba\!\_\thinspace cl}(x) =\max\Bigl(\big\{N^i_j\ |\
(i,j)\textrm{-role is final}\big\}\Bigr)$ \bigskip

It is left to capture all possible counteractions of type B(b).
Let:
\medskip
\begin{eqnarray*}
\begin{array}{ll}
    \mathcal{L}(x)^{W_1} := & \bigl\{W{r'}R.C\in \mathcal{L}(x)\ |
\textrm{ there are } k,\ k\geq 1,\ W \textrm{ part restrictions
} \smallskip\\
    & W{r}S.E\in \mathcal{L}(x),~\textrm{such that}~ \sum _{k} r +
r' = 1\bigr\}
  \end{array}
\end{eqnarray*}


$Den_1(x):=\bigl\{t\ |\ WrR.C \in \mathcal{L}(x)^{W_1},
~\textrm{and}~ r={s\over t}\bigr\}$ \medskip

\noindent i.e., $Den_1(x)$ contains the denominators of $r$-s in
the tuples of $W$ part restrictions in $\mathcal{L}(x)$  (with any
role and any concept) which can potentially B(b)-counteract. We
denote with $LCM_1(x)$ the least common multiple of the elements
of $Den_1(x)$, if $Den_1(x)\not = \emptyset$. Otherwise
$LCM_1(x)=1$.\medskip

We finally set:\medskip

$BAN_x=BAN_{A\_Ba\!\_\thinspace cl}(x)\times LCM_1(x)$
\medskip

The above considerations, together with Lemmas \ref{lemma_*-sat},
\ref{lemma_*-stab}, and \ref{lemma_Bb},
prove the lemma. $\circlearrowleft$
\end{proof}

Note that when the $BAN_x$ number of neighbours of $x$ on a role
in $\textbf{R}^P_x$ is reached in the application of some
generating rule, the part restrictions with that role may be
not---deterministically, with just that labelling of the
neighbours---c.t.-satisfied. But the lemma ensures that at that
moment \emph{there is} a labelling (which can be
non-deterministically chosen) c.t.-satisfying all other concepts
already c.t.-satisfied together with all part restrictions with
that role which are simultaneously c.t.-satisfiable, so there
would not be (an avoidable) clash with them. After that the only
possible clashes with part restrictions can come from part
restrictions which are not simultaneously c.t.-satisfiable (i.e.,
some of the inequalities $1^\circ$---$4^\circ$ fails for them). As
these clashes are unavoidable, no generation of new successors of
$x$ is necessary.\medskip

Let $|C|$ denote the length of the concept $C$, i.e., the number
of symbols in it, and let $|D|=n_0$, where $D$ is the concept
tested for satisfiability. Then for the set of all sub-concepts of
$D,\ sub(D)$, it holds $|sub(D)|\leq n_0$, and $|clos(D)|\leq
2n_0$, where $clos(D)=sub(D)\cup\big\{\!\!\sim\!C\ |\ C\in
sub(D)\big\}$.
\medskip

So we have for any $n$ in a concept in $clos(D)$, in case the
usual binary coding of numbers is used:
\smallskip

$n\leq 2^{n_0}$, and $\bar{n}\leq 2n_0.2^{n_0}$ \medskip

Also, each $r$ is a rational number practically in $({1\over
2^{n_0}},1-{1\over 2^{n_0}})$, and if ${\sum \bar{r}} < 1$ then
${\sum \bar{r}} \leq 1-{1 \over 2^{n_0}}$. So:
\smallskip

${1 \over {1-{\sum \bar{r}}}} \leq 2^{n_0}$ \medskip

As we have no more than $2n_0$ initial roles, as defined above,
and no more than $2n_0$ inclusions of roles from $\textbf{R}^P_x$
are possible, it is easy to show inductively that for any final
$i,j$-role $N^i_j\leq (2n_0)^{2n_0}.2^{(2n_0)^2}$. Thus \smallskip

$BAN_{A\_Ba\!\_\thinspace cl}(x) \leq
(2n_0)^{2n_0}.2^{2n_0^2+3n_0}$, and also \smallskip

$LCM_1(x) \leq 2^{2n_0^2}$, so we finally obtain: \smallskip

$BAN_x \leq (2n_0)^{2n_0}.2^{2n_0^2+3n_0}.2^{2n_0^2} =
(2n_0)^{2n_0}.2^{4 {n_0}^2 + 3n_0}$ \bigskip

In the following we will need the definition of \emph{pairwise
blocking}.

A node $x$ is \emph{directly blocked} if none of its ancestors is
blocked, and it has ancestors $x',\ y$ and $y'$ such that:

1. $x$ is a successor of $x'$, and $y$ is a successor of $y'$,

2. $\mathcal{L}(x)=\mathcal{L}(y)$ and
$\mathcal{L}(x')=\mathcal{L}(y')$, and

3. $\mathcal{L}(\langle x',x\rangle)=\mathcal{L}(\langle
y',y\rangle)$.

\noindent In this case we say that $y$ blocks $x$.

A node is \emph{indirectly blocked} if its predecessor is blocked.
Also, to avoid the expansion of a node $y$ witch is a successor of
a node $x$ and is merged by $\leqslant$-rule with
$\mathcal{L}(\langle x,y \rangle)=\emptyset$, $y$ is considered to
be indirectly blocked.

\begin{figure} 

\begin{tabular}{l l}
$\sqcap$-rule:\hspace{5.5mm} &If 1. $C_1\sqcap C_2\in
\mathcal{L}(x)$, $x$ is not indirectly blocked, and\\
&\hspace{2.2mm} 2. $\{C_1,C_2\}\not\subseteq\mathcal{L}(x)$ \\
&then $\mathcal{L}(x)\rightarrow\mathcal{L}(x)\cup\{C_1,C_2\}$ \medskip \\

$\sqcup$-rule: &If 1. $C_1\sqcup C_2\in
\mathcal{L}(x)$, $x$ is not indirectly blocked, and\\
&\hspace{2.2mm} 2. $\{C_1,C_2\}\cap\mathcal{L}(x)=\emptyset$ \\
&then $\mathcal{L}(x)\rightarrow\mathcal{L}(x)\cup\{E\}$ for some
$E\in\{C_1,C_2\}$ \medskip \\

$\exists$-rule: &If 1. $\exists R.C\in
\mathcal{L}(x)$, $x$ is not blocked, and\\
&\hspace{2.2mm} 2. $x$ has no $R$-neighbour $y$ with $C\in\mathcal{L}(y)$\\
&then create a new successor $y$ of $x$ with $\mathcal{L}(\langle
x,y\rangle)=\{ R\}$ and $\mathcal{L}(y)=\{C\}$ \medskip \\

$\forall$-rule: &If 1. $\forall R.C\in
\mathcal{L}(x)$, $x$ is not indirectly blocked, and\\
&\hspace{2.2mm} 2. there is an $R$-neighbour $y$ of $x$ with $C\not\in\mathcal{L}(y)$\\
&then $\mathcal{L}(y)\rightarrow\mathcal{L}(y)\cup\{C\}$ \medskip \\

$\forall_+$-rule: &If 1. $\forall S.C\in
\mathcal{L}(x)$, $x$ is not indirectly blocked, \\
&\hspace{2.2mm} 2. there is some transitive $R$ with $R\sqsubseteq^+ S$, and\\
&\hspace{2.2mm} 3. there is an $R$-neighbour $y$ of $x$ with $\forall R.C\not\in\mathcal{L}(y)$\\
&then $\mathcal{L}(y)\rightarrow\mathcal{L}(y)\cup\{\forall R.C\}$ \medskip \\

${choose}$- &If 1. $\boxtimes R.C\in \mathcal{L}(x)$, $x$ is not
indirectly blocked, and\\
rule: &\hspace{2.2mm} 2. there is an $R$-neighbour $y$ of $x$
with $\{C,\sim\! C\}\cap\mathcal{L}(y)=\emptyset$\\
&then $\mathcal{L}(y)\rightarrow\mathcal{L}(y)\cup\{E \}$ for some
$E\in\{C,\sim\! C\}$ \medskip \\

$\geqslant$-rule: &If 1. $\geqslant nR.C\in \mathcal{L}(x)$,
$x$ is not blocked, and\\
&\hspace{2.3mm} 2. $\sharp  R^\textbf{T}(x,C) < n$\\
&then create a new successor $y$ of $x$ with $\mathcal{L}(\langle
x,y\rangle)=\{ R\}$ and $\mathcal{L}(y)=\{C\}$ \medskip \\

$\leqslant$-rule: &If 1. $\leqslant nR.C\in \mathcal{L}(x)$,
$x$ is not indirectly blocked, and\\
&\hspace{2.3mm} 2. $\sharp  R^\textbf{T}(x,C) > n$ and there are
two $R$-neighbours $y$ and $z$ of $x$ with\\
&\hspace{6mm} $C\in\mathcal{L}(y)$, $C\in\mathcal{L}(z)$,
$\mathcal{L}(\langle x,y\rangle)\cap\mathcal{L}(\langle x,z\rangle)=\emptyset$, and\\
&\hspace{6mm} $y$ is not a predecessor of $x$ \\
& then 1.
$\mathcal{L}(z)\rightarrow\mathcal{L}(z)\cup\mathcal{L}(y)$,\\
&\hspace{6.3mm} 2. If $z$ is a predecessor of $x$\\
&\hspace{10.3mm} then $\mathcal{L}(\langle
z,x\rangle)\rightarrow\mathcal{L}(\langle z,x\rangle)\cup
\{\textsf{Inv}(S)|S\in\mathcal{L}(\langle x,y\rangle)\}$\\
&\hspace{10.3mm} else $\mathcal{L}(\langle
x,z\rangle)\rightarrow\mathcal{L}(\langle x,z\rangle)\cup
\mathcal{L}(\langle x,y\rangle)$\\
&\hspace{6.3mm} 3. $\mathcal{L}(\langle
x,y\rangle)\rightarrow\emptyset$ \medskip \\

${M}$-rule: &If 1. $MrR.C\in \mathcal{L}(x)$, $x$
is not blocked, and\\
&\hspace{2.2mm} 2. $\sharp  R^\textbf{T}(x,C)\leq r.\sharp
R^\textbf{T}(x)$ then calculate $BAN_x$, and if\\
&\hspace{2.2mm} 3. $\sharp  R^\textbf{T}(x)<BAN_x$\\
&then create a new successor $y$ of $x$ with $\mathcal{L}(\langle
x,y\rangle)=\{ R\}$ and $\mathcal{L}(y)=\{ C\}$\medskip \\

${W}$-rule: &If 1. $WrR.C\in \mathcal{L}(x)$, $x$
is not blocked, and\\
&\hspace{2.2mm} 2. $\sharp  R^\textbf{T}(x,\sim\! C) > r.\sharp
R^\textbf{T}(x)$ then calculate $BAN_x$, and if\\
&\hspace{2.2mm} 3. $\sharp  R^\textbf{T}(x)<BAN_x$\\
&then create a new successor $y$ of $x$ with $\mathcal{L}(\langle
x,y\rangle)=\{ R\}$ and $\mathcal{L}(y)=\{ C\}$
\end{tabular}

    \caption{The set of completion rules}\label{figure:rules_all}
\end{figure}

\medskip

\noindent \textbf{Theorem~1.} \emph{The presented tableaux
algorithm is a decision procedure for the satisfiability and
subsumption of $\mathcal{ALCQPIH}_{R^+}$-concepts with respect to
role hierarchies and terminologies.}

\begin{proof}
The correctness of the algorithm as a decision procedure is shown
by proving the lemmas that it is terminating, sound, and complete.

The termination is ensured by pair-wise blocking, and by
$BAN$-checkup, which guarantees finite (at most exponential---in
case of the usual binary coding of numbers) branching at a node.

Part restrictions can only affect the c.t.-satisfaction of `at
most' qualifying number restrictions in a node label. From the
other side, the influence of other concepts on part restrictions
is captured in the calculation of $BAN_x$. Thus, to show soundness
and completeness it suffices to consider, additionally to the
proofs, presented in Lemmas 16 and 17 in section 5.4 of
\cite{HST99}, the cases with the qualifying number restrictions
and part restrictions, and the cases with the new and modified
completion rules.

Since the \emph{internalization of terminologies} \cite{Baa90} is
still possible in the presence of part restrictions following the
technique presented in Section 3.1 of \cite{HST99}, we obtain the
above result. $\circlearrowleft$
\end{proof}

\begin{lemma}[Termination] \label{lemma:termination}
For each $\mathcal{ALCQPIH}_{R^+}$-concept $D$ and role hierarchy
$\mathcal{R}^+$, the tableaux algorithm terminates.
\end{lemma}

\begin{proof}
The modified $\geqslant$-rule and $\leqslant$-rule and part rules
keep all properties of the completion rules which are shown in
Lemma 15 in \cite{HST99} to guarantee the termination of tableaux
algorithm for any $\mathcal{ALCQIH}_{R^+}$-concept. Namely, any
sequence of rule applications is finite, as:

\begin{enumerate}
    \item The completion rules never remove nodes from the tree, and also
    never remove concepts from node labels. Edge labels can be changed
    only by the $\leqslant$-rule which only transfers roles between edge
    labels of merged nodes. In that, the node below the
    $\emptyset$-labelled edge becomes blocked.\smallskip

    \item If a part rule is applied for a part restriction,
    labelling some node $x$, it adds no more than $BAN_x$ successors
    of $x$ for a fixed role. With the notations and considerations from the end
    of Lemma \ref{lemma:BAN}, $BAN_x$ is limited by $(2n_0)^{2n_0}.2^{4 {n_0}^2 + 3n_0}$,
    and the number of roles is limited by $2n_0$. So, the out-degree of the tree, caused
    by part restrictions (and also the overall out-degree of the tree, as in the
    calculation of $BAN_x$ the successors, possibly added by the $\exists$-rule
    and $\geqslant$-rule are taken into account) is bounded by $\mathcal{O}(2^{n_0^2})$.
    \smallskip

    \item Part rules do not affect in a special way the paths in
    the tree (but just in the way the $\geqslant$-rule does),
    so the considerations, showing that paths are of length
    at most $2^{4n_0k}$, where $k=|\textbf{R}_D| \leq 2n_0$, still hold.
$\circlearrowleft$\bigskip
\end{enumerate}
\end{proof}

\begin{lemma}[Soundness] \label{lemma:sound}
If the completion rules can be applied to a concept $D$ of
$\mathcal{ALCQPIH}_{R^+}$ to yield a complete and clash-free
completion tree with respect to the role hierarchy
$\mathcal{R}^+$, then $D$ has a tableau with respect to
$\mathcal{R}^+$.
\end{lemma}

\begin{proof}
We use the path construction considered in Lemma 16 in
\cite{HST99} to cope with the qualifying number restrictions, and
show that it deals successfully with part restrictions also.

Let $\textbf{T}$ be a complete and clash-free completion tree. We
consider a \emph{path} to be a sequence of pairs of nodes of
$\textbf{T}$ of the form $p=[(x_0,x_0'),...,(x_n,x_n')]$. For such
a path we define two functions: $\textsf{Tail}(p):=x_n$ and
$\textsf{Tail}'(p):=x_n'$. We use also $[p|(x_{n+1},x_{n+1}')]$ to
denote the above path $p$ extended with the pair of nodes
$(x_{n+1},x_{n+1}')$. Then we define the set
$\textsf{Paths}(\textbf{T})$ inductively:
\begin{itemize}
    \item [$\bullet$]
    For the root $x_0$ of $\textbf{T}$, $[(x_0,x_0)]\in \textsf{Paths}
    (\textbf{T})$, and \medskip

    \item [$\bullet$]
    For a path $p\in\textsf{Paths}(\textbf{T})$ and a node
    $z$ in $\textbf{T}$ \smallskip
    \begin{itemize}
        \item [--]
        if $z$ is a successor of $\textsf{Tail}(p)$ and $z$
        is not blocked, then $[p|(z,z)]\in \textsf{Paths}
        (\textbf{T})$, or \smallskip

        \item [--]
        if, for some node $y$ in $\textbf{T}$, $y$ is a
        successor of $\textsf{Tail}(p)$ and $z$
        blocks $y$, then $[p|(z,y)]\in\textsf{Paths}(\textbf{T})$.
    \end{itemize}
\end{itemize}

Thus, by the construction, if $p\in\textsf{Paths}(\textbf{T})$ and
$p=[p'|(x,x')]$, then $x$ is not blocked, $x'$ is blocked iff
$x'\not = x$, and $x'$ is not indirectly blocked. Also,
$\mathcal{L}(x)=\mathcal{L}(x')$ always holds.  \medskip

Now we define a tableau  $T=({\textbf{S}}, \mathcal{L},
\mathcal{E})$ as follows:
\begin{eqnarray*}
  \textbf{S} &=& \textsf{Paths}(\textbf{T}) \\
  \mathcal{L}(p) &=& \mathcal{L}(\textsf{Tail}(p)) \\
  \mathcal{E}(R) &=& \big \{\langle p,q\rangle\in \textbf{S}\times
  \textbf{S}\ \bigg |\
  \begin{array}{l}
    \textrm{either}\ q=[p|(x,x')]\ \textrm{and }
    x'\textrm{ is an }R\textrm{-successor of }\textsf{Tail}(p)\smallskip\\
    \textrm{or}\ p=[q|(x,x')]\ \textrm{and }
    x'\textrm{ is an } \textsf{Inv}(R)\textrm{-successor of }\textsf{Tail}(q)\\
  \end{array}
  \big \}
\end{eqnarray*}

In order to show that $T$ is a tableau for $D$ with respect to
$\mathcal{R}^+$ we have to show that $T$ satisfies the properties
1--13 from the tableaux definition. The proofs for properties 1--8
are exactly the same as those given in Lemma 16 in \cite{HST99}.
We complete the proof by showing that $T$ satisfies properties
9--13 which are connected with the modified and the new completion
rules. To prove in a uniform way that properties 9, 10, 11 and 12
hold, we use the following proposition.

\begin{proposition} \label{proposition:num}
Let $p\in\emph{\textsf{Paths}}(\emph{\textbf{T}})$, $R\in
\emph{\textbf{R}}_D$, and $C\in clos(D)$. Then it holds $\sharp
R^T(p,C)= \sharp R^\emph{\textbf{T}}(\emph{\textsf{Tail}}(p),C)$.
\end{proposition}

\begin{proof} Let $\textsf{Tail}(p)=x$. We consider several sets:
\begin{eqnarray*}
    Succ^p &=& \{q\in R^T(p,C)\ |\ q=[p|(y,y')]\} \\
    Pred^p &=& \{q\in R^T(p,C)\ |\ p=[q|(x,x')]\} \\
    Succ^x &=& \{y\in R^{\textbf{T}}(x,C)\ |\ y\textrm{ is }
    R\textrm{-successor of }x\} \\
    Pred^x &=& \{y\in R^{\textbf{T}}(x,C)\ |\ x\textrm{ is }
  \textsf{Inv}(R)\textrm{-successor of }y\}
\end{eqnarray*}

As $R^T(p,C)=Succ^p\cup Pred^p$, and $Succ^p\cap
Pred^p=\emptyset$, it holds
\[ \sharp R^T(p,C)= \sharp Succ^p + \sharp Pred^p
\]
In the same way $R^\textbf{T}(x,C)=Succ^x\cup Pred^x$, and
$Succ^x\cap Pred^x=\emptyset$, so that
\[ \sharp R^\textbf{T}(x,C)= \sharp Succ^x + \sharp Pred^x
\]

$Pred^x$ is either empty, or a singleton, and due to the
construction of paths and the definition of $\mathcal{E}(R)$, so
is $Pred^p$. Let $p=[q|(x,x')]$ and $\textsf{Tail}(q)=z$.\medskip

Let $Pred^x\not = \emptyset$ with $y\in Pred^x$. Then $x$ is
$\textsf{Inv}(R)$-successor of $y$ and $C\in \mathcal{L}(y)$. We
will show that $q\in Pred^p$.
\begin{itemize}
    \item If $x=x'$ then $x$ is a successor of $z$, so
$y=z$, $x$ is $\textsf{Inv}(R)$-successor of $z,\ C\in
\mathcal{L}(z)$, and $q\in Pred^p$.

    \item If $x\not =x'$ then $x'$ is blocked by
$x$. As $x'$ is a successor of $z$, from the blocking conditions
it follows that $x'$ is $\textsf{Inv}(R)$-successor of $z,\ C\in
\mathcal{L}(z)=\mathcal{L}(y)$, and again $q\in Pred^p$.
\end{itemize}

Now let $q\in Pred^p$. Then $x'$ is $\textsf{Inv}(R)$-successor of
$z$ and $C\in \mathcal{L}(z)$. We will show for the predecessor
$y$ of $x$ that $y\in Pred^x$.
\begin{itemize}
    \item If $x=x'$ then $x$ is $\textsf{Inv}(R)$-successor of $z$, so
$y=z$, $x$ is $\textsf{Inv}(R)$-successor of $y,\ C\in
\mathcal{L}(y)$, and $y\in Pred^x$.

    \item If $x\not =x'$ then $x'$ is blocked by
$x$. As $x'$ is $\textsf{Inv}(R)$-successor of $z$, from the
blocking conditions it follows that $x$ is
$\textsf{Inv}(R)$-successor of $y,\ C\in
\mathcal{L}(y)=\mathcal{L}(z)$, and again $y\in Pred^x$.
\end{itemize}

We showed that $Pred^x\not=\emptyset$ iff $Pred^p\not=\emptyset$,
so that $\sharp Pred^x=\sharp Pred^p$.\medskip

It is left to show that $\sharp Succ^x=\sharp Succ^p$.\smallskip

If $y\in Succ^x$ then $q=[p|(y',y)]$, where either $y'=y$, or $y'$
blocks $y$, is in $Succ^p$, as in both cases
$C\in\mathcal{L}(q)=\mathcal{L}(y')=\mathcal{L}(y)$. Also, if
$y_1\not=y_2$ are in $Succ^x$ then $q_1=[p|(y',y_1)]$ and
$q_2=[p|(y'',y_2)]$ are different elements of $Succ^p$ even if
$y_1$ and $y_2$ are both blocked by the same node, i.e., $y'=y''$.
So $\sharp Succ^x \leq \sharp Succ^p$.\medskip

Now let $q=[p|(y,y')]\in Succ^p$. Then $\textsf{Tail'}(q)=y'$ is
$R$-successor of $x=\textsf{Tail}(p)$, and $C \in \mathcal{L}(y')=
\mathcal{L}(y)= \mathcal{L}(q)$, so $y'\in Succ^x$. We will show
that $\textsf{Tail}'$ is injective on $Succ^p$. Suppose, there are
$q_1\not=q_2$ in $Succ^p$ such that $q_1=[p|(y_1,y')],\
q_2=[p|(y_2,y')]$. If $y'$ is not blocked than $y_1=y'=y_2$, what
contradicts with $q_1\not=q_2$. But if $y'$ is blocked than again
$y_1=y_2$, as each of them has to block $y'$ and there is just one
blocking node for any node---again contradiction with
$q_1\not=q_2$. So for any element in $Succ^p$ there is a different
element in $Succ^x$, and $\sharp Succ^x = \sharp Succ^p$
$\circlearrowleft$.
\medskip

\end{proof}

\begin{corollary} \label{corollary:num}
Let $p\in\emph{\textsf{Paths}}(\emph{\textbf{T}})$, and $R\in
\emph{\textbf{R}}_D$. Then $\sharp R^T(p)= \sharp
R^\emph{\textbf{T}}(\emph{\textsf{Tail}}(p))$.
\end{corollary}

\begin{proof}
In the trivial case $\sharp R^T(p)= \sharp
R^{\textbf{T}}({\textsf{Tail}}(p))=0$.

In the non-trivial case, we can consider, as in the proposition
\ref{proposition:num}, the sub-sets of predecessors and of
successors on $R$ of both $p$ and $\textsf{Tail}(p)$. Both
sub-sets of predecessors are at the same time empty, or
singletons.

The sub-sets of successors are again at the same time empty, or
not. If $\textsf{Tail}(p)$ has at least one successor, then this
successor is generated by some generating rule triggered by a
concept of the form $\boxtimes R.C\in
\mathcal{L}(\textsf{Tail}(p))$. As $\textbf{T}$ is complete and
clash-free, due to $choose$-rule any neighbour of
$\textsf{Tail}(p)$ is labelled with either $C$, or $\sim\! C$, so
we can use the proposition \ref{proposition:num} to obtain what we
want. $\circlearrowleft$
\end{proof}

\begin{itemize}

    \item [$\bullet$] \textbf{Property 9}. Let $\leqslant\! nR.C\in\mathcal{L}(p)
    =\mathcal{L}(\textsf{Tail}(p))$. As $\textbf{T}$
    is clash-free, it holds $\sharp R^\textbf{T}(\textsf{Tail}(p),C)
    \leq n$. From Proposition \ref{proposition:num} we obtain $\sharp R^T(p,C) =
    \sharp R^\textbf{T}(\textsf{Tail}(p),C) \leq n$. \medskip

    \item [$\bullet$] \textbf{Property 10}. Let $\geqslant nR.C\in\mathcal{L}(p)
    =\mathcal{L}(\textsf{Tail}(p))$. As $\textbf{T}$
    is complete, $\geqslant$-rule is inapplicable, so
    $\sharp R^\textbf{T}(\textsf{Tail}(p),C)
    \geq n$. From Proposition \ref{proposition:num} $\sharp R^T(p,C) =
    \sharp R^\textbf{T}(\textsf{Tail}(p),C) \geq n$. \medskip

    \item [$\bullet$] \textbf{Property 11}. Let $MrR.C\in
    \mathcal{L}(p)=\mathcal{L}(\textsf{Tail}(p))$. As $\textbf{T}$
    is clash-free, it holds $\sharp R^\textbf{T}(\textsf{Tail}(p),C)
    > r.\sharp R^\textbf{T}(\textsf{Tail}(p))$. Making use of
    Proposition \ref{proposition:num} and Corollary
    \ref{corollary:num}, we obtain $\sharp R^T(p,C) =
    \sharp R^\textbf{T}(\textsf{Tail}(p),C) >
    r.\sharp R^\textbf{T}(\textsf{Tail}(p)) = r.\sharp R^T(p)$. \medskip

    \item [$\bullet$] \textbf{Property 12}. Let $WrR.C\in
    \mathcal{L}(p)=\mathcal{L}(\textsf{Tail}(p))$. As $\textbf{T}$
    is clash-free, it holds $\sharp R^\textbf{T}(\textsf{Tail}(p),\sim\! C)
    \leq r.\sharp R^\textbf{T}(\textsf{Tail}(p))$. Again, using
    Proposition \ref{proposition:num} and Corollary
    \ref{corollary:num}, we obtain $\sharp R^T(p,\sim\! C) =
    \sharp R^\textbf{T}(\textsf{Tail}(p),\sim\! C) \leq
    r.\sharp R^\textbf{T}(\textsf{Tail}(p)) = r.\sharp R^T(p)$. \medskip

    \item [$\bullet$] \textbf{Property 13}. Let $\boxtimes R.C \in
    \mathcal{L}(p)=\mathcal{L}(\textsf{Tail}(p))=\mathcal{L}(\textsf{Tail'}(p))$,
    and $\langle p,q\rangle\in\mathcal{E}(R)$. If
    $q=[p|(x,x')]$ then $x'\textrm{ is an }R\textrm{-successor of
    }\textsf{Tail}(p)$, and, as the $choose$-rule is not applicable,
    it holds $\{C,\sim\! C\}\cap \mathcal{L}(x')\not = \emptyset$.
    Since $\mathcal{L}(q)=\mathcal{L}(x)=\mathcal{L}(x')$, it also holds
    $\{C,\sim\! C\}\cap \mathcal{L}(q)\not =
    \emptyset$. If $p=[q|(x,x')]$ then
    $x'\textrm{ is an } \textsf{Inv}(R)$-successor of
    $\textsf{Tail}(q)$. As $x'$ is not indirectly blocked and the
    $choose$-rule is not applicable, $\{C,\sim\! C\}\cap \mathcal{L}
    (\textsf{Tail}(q))\not = \emptyset$ holds. Thus, $\{C,\sim\! C\}
    \cap \mathcal{L}(q)\not = \emptyset$ holds again.
$\circlearrowleft$ \medskip

\end{itemize}

\end{proof}

\begin{lemma}[Completeness] \label{lemma:complete}
If an $\mathcal{ALCQPIH}_{R^+}$-concept $D$ has a tableau with
respect to the role hierarchy $\mathcal{R}^+$, then the completion
rules can be applied to $D$ in a way to yield a complete and
clash-free completion tree with respect to $\mathcal{R}^+$.
\end{lemma}

\begin{proof}
Let $T=({\textbf{S}}, \mathcal{L}, \mathcal{E})$ be a tableau for
$D$ with respect to $\mathcal{R}^+$. The tableau is used to guide
the application of the non-deterministic rules. For that, a
function $\pi$ mapping the nodes of a tree $\textbf{T}$ to
$\textbf{S}$ is defined inductively such that, for each $x,y$ in
$\textbf{T}$, the following holds:
\begin{eqnarray*}
\begin{array}{ll}
\begin{array}{l}
  \mathcal{L}(x)\subseteq \mathcal{L}(\pi(x)) \smallskip\\
  \textrm{if }y\textrm{ is an }R\textrm{-neighbour of }x \textrm{ then }
  \langle\pi(x),\pi(y)\rangle \in \mathcal{E}(R) \smallskip\\
  \textrm{if }y\textrm{ and }z\textrm{ are neighbours of }x\textrm{ and }
  \mathcal{L}(\langle x,y\rangle)\cap\mathcal{L}(\langle
  x,z\rangle)\neq\emptyset\\
  \textrm{then }\pi(y)\not = \pi(z) \\
\end{array} &\hspace{2mm}\Biggr \}\ (*)
\end{array}
\end{eqnarray*}

Let $\textbf{T}$ be a completion tree and $\pi$ be a function that
satisfies $(*)$. The claim is that if a rule is applicable to
$\textbf{T}$ then it is applicable in a way that yields a
completion tree $\textbf{T}'$ and a function $\pi'$ that satisfy
$(*)$. We will show that for the new and modified rules.

\begin{itemize}

    \item [$\bullet$] $choose$-rule. Let $\boxtimes R.C \in
    \mathcal{L}(x)$, and the $choose$-rule is applicable to $x$
    for the $R$-neighbour $y$ of $x$. So $\{C,\sim\! C\}\cap \mathcal{L}(y) = \emptyset$.
    From $(*)$
    $\boxtimes R.C \in \mathcal{L}(\pi(x))$ and $\langle\pi(x),\pi(y)\rangle
    \in\mathcal{E}(R)$. Then the Property 13. from the tableaux
    definition implies $\{C,\sim\! C\}\cap \mathcal{L}(\pi(y))\not =
    \emptyset$. Thus, $choos$-rule can add the concept from $\{C,\sim\! C\}$
    being in $\mathcal{L}(\pi(y))$ to $\mathcal{L}(y)$, so that $\mathcal{L}(y) \subseteq
    \mathcal{L}(\pi(y))$ holds. \medskip

    \item [$\bullet$] $\geqslant$-rule. Let $\geqslant\!\!nR.C\in\mathcal{L}(x)\subseteq
    \mathcal{L}(\pi(x))$, and the $\geqslant$-rule is applicable to $x$.
    So, $\sharp R^{\textbf{T}}(x,C) < n$.
    From the other side, from the tableau we have $\sharp R^{T}(\pi(x),C) \geq
    n$.

    From $\sharp R^{T}(\pi(x),C) \geq n > \sharp R^{\textbf{T}}(x,C)$ it follows that
    there is some $t_0 \in R^{T}(\pi(x),C)$ which is not a $\pi$-image of
    any $y \in R^{\textbf{T}}(x,C)$. Suppose $x$ has an $R$-neighbour
    $y_0 \not\in R^{\textbf{T}}(x,C)$, such that $t_0=\pi(y_0)$. Due to the rule
    application strategy, the $choose$-rule has been applied
    before the $\geqslant$-rule becomes applicable, and $C \not\in \mathcal{L}(y_0)$,
    so $\sim\! C \in \mathcal{L}(y_0)\subseteq \mathcal{L}(t_0)$. As
    $C \in \mathcal{L}(t_0)$, the latter contradicts with Property 1
    of tableaux.

    So, $t_0$ is not a $\pi$-image of any $y\in
    R^{\textbf{T}}(x)$, and $\geqslant$-rule can add new $R$-successor $y_0$ of $x$ with
    $\mathcal{L}(y_0)=\{ C\}$, and $\pi$ can be extended to map $y_0$ to $t_0$,
    so that $(*)$ holds. \medskip

    \item [$\bullet$] $\leqslant$-rule. Let
    $\leqslant nR.C\in\mathcal{L}(x)\subseteq
    \mathcal{L}(\pi(x))$, and the $\leqslant$-rule is applicable to $x$.
    So, $\sharp R^{\textbf{T}}(x,C) > n$.
    From the tableau we have $\sharp R^{T}(\pi(x),C) \leq n$.

    From $\sharp R^{T}(\pi(x),C) \leq n < \sharp R^{\textbf{T}}(x,C)$ it follows that
    there are at least two nodes $y, z \in R^{\textbf{T}}(x,C)$ for which
    $\pi(y)=\pi(z)$, and at least one of them is an $R$-successor of $x$, let it be
    $y$. As $(*)$ holds, $\mathcal{L}(\langle x,y\rangle)\cap\mathcal{L}(\langle
    x,z\rangle)=\emptyset$ $\big($so there is a role $S\sqsubseteq^+R$ such that
    $S\in\mathcal{L}(\langle x,y\rangle)$, or either $S\in\mathcal{L}(\langle
    x,z\rangle)$ or $\textsf{Inv}(S)\in\mathcal{L}(\langle
    z,x\rangle)$$\big)$. Then the $\leqslant$-rule can be applied to
    merge $y$ and $z$ into $z$ without violating $(*)$.\medskip

    \item [$\bullet$] $M$-rule. Let $MrR.C\in\mathcal{L}(x)\subseteq \mathcal{L}(\pi(x))$,
    and the $M$-rule is applicable to $x$. So, $\sharp R^{\textbf{T}}(x,C)\leq
    r.\sharp R^{\textbf{T}}(x)$. From the other side, from the tableau we
    have $\sharp R^{T}(\pi(x),C) > r.\sharp R^{T}(\pi(x))$.

    From $(*)$ $\pi$ is an injection from $R^{\textbf{T}}(x)$ to $R^{T}(\pi(x))$,
    and from $R^{\textbf{T}}(x,C)$ to $R^{T}(\pi(x),C)$. So, $\sharp R^{\textbf{T}}(x)
    \leq \sharp R^{T}(\pi(x))$. Thus, we obtain $\sharp R^{\textbf{T}}(x,C) < \sharp
    R^{T}(\pi(x),C)$. This shows that there is an individual
    $s\in R^{T}(\pi(x),C)$ for which there is not a node $y\in
    R^{\textbf{T}}(x,C)$ such that $s=\pi(y)$.

    Suppose, there is $y\in
    R^{\textbf{T}}(x)$ such that $s=\pi(y)$. But, as $M$-rule is
    applicable for $x$, $choose$-rule is inapplicable, and $C\not\in
    \mathcal{L}(y)$, so $\sim\!
    C\in\mathcal{L}(y)\subseteq\mathcal{L}(\pi(y))=\mathcal{L}(s)$.
    This contradicts with Property 1 from tableaux definition, as
    $C\in\mathcal{L}(s)$. So, there is no $R$-neighbour $y$ of $x$ with
    $s=\pi(y)$, and $M$-rule can be applied to generate a new
    $R$-successor $y$ of $x$ with $\mathcal{L}(y)=\{C\}$. Then, the extension
    of $\pi$ which maps $y$ to $s$ holds $(*)$.\medskip

    \item [$\bullet$] $W$-rule. Let $WrR.C\in\mathcal{L}(x)\subseteq \mathcal{L}(\pi(x))$,
    and the $W$-rule is applicable to $x$. So, $\sharp R^{\textbf{T}}(x,\sim\!
    C) > r.\sharp R^{\textbf{T}}(x)$. From the other side, from the tableau we
    have $\sharp R^{T}(\pi(x),\sim\! C) \leq r.\sharp R^{T}(\pi(x))$.

    Note that these inequalities can be rewritten,
    due to the inapplicability of the $choose$-rule, the
    inclusion of $\mathcal{L}(x)$ in $\mathcal{L}(\pi(x))$,
    Property 1, and Property 13 of the tableau, as $\sharp R^{\textbf{T}}(x,C) <
    (1-r).\sharp R^{\textbf{T}}(x)$ and $\sharp R^{T}(\pi(x),C) \geq
    (1-r).\sharp R^{T}(\pi(x))$. Next, the same arguments as for the
    $M$-rule hold.\medskip
\end{itemize}
Since $T$ is a tableau for $D$, there is  an $s_0\in \textbf{S}$
with $D\in\mathcal{L}(s_0)$. So, the above claim gives the
completeness, as the initial tree consisting of a single node
$x_0$ with $\mathcal{L}(x_0)=\{D\}$, together with the function
$\pi:\ \pi(x_0):=s_0$, hold $(*)$. Then the rules, applied in a
way to preserve $(*)$, will yield a complete (as any sequence of
rule application is finite) and clash-free completion tree. This
completes the proof of Theorem $\ref{main-DP}$ $\circlearrowleft$

\end{proof}

\end{document}